%% file: main.tex
\documentclass[acmsmall,nonacm,screen,10pt]{acmart}
\AtBeginDocument{%
  \providecommand\BibTeX{{%
   \normalfont B\kern-0.5em{\scshape i\kern-0.25em b}\kern-0.8em\TeX}}}

\usepackage{booktabs}
\usepackage{subcaption}
\usepackage{mdframed}
\usepackage[bottom]{footmisc}

\usepackage{balance}
\usepackage{pifont}
\usepackage{wrapfig}
\usepackage{hyperref}
\usepackage{fancyvrb}
\usepackage{xspace}

\usepackage[ruled, linesnumbered]{algorithm2e}
\usepackage{amsfonts}
\usepackage{amsmath}
\usepackage{amsthm}
\usepackage{semantic}
\usepackage{listings}
\usepackage{array}
\usepackage{fancyvrb}
\usepackage{mathpartir}
\usepackage{stmaryrd}
\usepackage{graphicx}
\usepackage{caption}
\usepackage{syntax, etoolbox}
\usepackage{tikz}
\usepackage{proof}
\usepackage{mathtools}
\usepackage{adjustbox}
\usepackage{bbding}
\usepackage{soul}
\usepackage{multirow}
\usepackage[normalem]{ulem}
\usepackage{colortbl}

\usepackage{tabularx}

\usepackage{environ}
\usepackage{bcprules}
\usepackage{ebproof}

\usepackage{balance}
\usepackage{pifont}
\usepackage{wrapfig}
\usepackage{hyperref}
\usepackage{fancyvrb}
\usepackage{xspace}

\usepackage{amsfonts}
\usepackage{amsmath}
\usepackage{amsthm}
\usepackage{semantic}
\usepackage{listings}
\usepackage{array}
\usepackage{fancyvrb}
\usepackage{stmaryrd}
\usepackage{graphicx}
\usepackage{caption}
\usepackage{syntax, etoolbox}
\usepackage{tikz}
\usepackage{proof}
\usepackage{color}
\usepackage{mathtools}
\usepackage{adjustbox}
\usepackage{multirow}
\usepackage{bbding}
\usepackage{soul} %
\usepackage{multirow}
\usepackage{fontawesome}
\usepackage{pgfplots}
\pgfplotsset{compat=1.18}
\pgfplotsset{ log ticks with fixed point, }
\usetikzlibrary{plotmarks}
\usetikzlibrary{calc}

\newcommand{\symweight}{\theta}
\newcommand{\assignment}{w}

\newif\ifdraft\draftfalse{}

\newif\iflater\latertrue{}
\newif\ifcharts\chartstrue{}
\newif\ifappendix\appendixtrue{}
\newif\ifafterdeadline\afterdeadlinetrue{}

\ifcharts
    \pgfplotstableread[col sep = tab]{data/bst.csv}\bstdata
    \pgfplotstableread[col sep = tab]{data/rbt.csv}\rbtdata
    \pgfplotstableread[col sep = tab]{data/stlc.csv}\stlcdata
    \pgfplotstableread[col sep = tab]{data/bespokestlc.csv}\bespokestlcdata
    
    \pgfplotstableread[col sep = tab]{data/cunif.csv}\cumulativeuniqdata
    \pgfplotstableread[col sep = tab]{data/rbtablationcuniq.csv}\rbtablationcumulativeuniqdata
    \pgfplotstableread[col sep = tab]{data/rbt-cuniq-size-only-bound00.csv}\rbtablationcumulativeuniqsizebounddata
    
    \pgfplotstableread[col sep = tab]{data/eight_stlc_bespoke_depth_initial.csv}\stlcbespokeinitdepth
    \pgfplotstableread[col sep = tab]{data/eight_stlc_bespoke_depth_uniform.csv}\stlcbespokeunifdepth
    \pgfplotstableread[col sep = tab]{data/eight_stlc_bespoke_depth_linear.csv}\stlcbespokelindepth
    
    \pgfplotstableread[col sep = tab]{data/eight_unifstlc.csv}\stlcuniftarget
    \pgfplotstableread[col sep = tab]{data/eight_linearstlc.csv}\stlclineartarget
    \pgfplotstableread[col sep = tab]{data/eight_uniftbstlc.csv}\stlctbuniftarget
    \pgfplotstableread[col sep = tab]{data/eight_lineartbstlc.csv}\stlctblineartarget
    \pgfplotstableread[col sep = tab]{data/eight_uniftree.csv}\treeuniftarget
    \pgfplotstableread[col sep = tab]{data/eight_lineartree.csv}\treelineartarget
    
    \pgfplotstableread[col sep = tab]{data/eight_so_bst_tb_depth_initial.csv}\bsttbinitdepth
    \pgfplotstableread[col sep = tab]{data/eight_so_bst_tb_depth_uniform.csv}\bsttbunifdepth
    \pgfplotstableread[col sep = tab]{data/eight_so_bst_tb_depth_linear.csv}\bsttblindepth
    \pgfplotstableread[col sep = tab]{data/eight_so_rbt_tb_depth_initial.csv}\rbttbinitdepth
    \pgfplotstableread[col sep = tab]{data/eight_so_rbt_tb_depth_uniform.csv}\rbttbunifdepth
    \pgfplotstableread[col sep = tab]{data/eight_so_rbt_tb_depth_linear.csv}\rbttblindepth
    \pgfplotstableread[col sep = tab]{data/eight_so_stlc_tb_depth_initial.csv}\stlctbinitdepth
    \pgfplotstableread[col sep = tab]{data/eight_so_stlc_tb_depth_uniform.csv}\stlctbunifdepth
    \pgfplotstableread[col sep = tab]{data/eight_so_stlc_tb_depth_linear.csv}\stlctblindepth
    
    \pgfplotstableread[col sep = tab]{data/app_initial.csv}\appsinitial
    \pgfplotstableread[col sep = tab]{data/app_target.csv}\appstarget
    \pgfplotstableread[col sep = tab]{data/app_trained.csv}\appstrained
\fi

\definecolor{dkred}{rgb}{0.7,0,0}
\definecolor{dkred}{rgb}{0,0,0.7}
\definecolor{myblue}{rgb}{0.3,0.5,1.0}
\definecolor{mydarkblue}{rgb}{0,0,1.0}
\definecolor{notered}{rgb}{0.85,0,0}
\definecolor{dkpurple}{HTML}{4e02eb}
\definecolor{dkgreen}{HTML}{006329}
\definecolor{dkorange}{HTML}{cc5500}

\usepackage{placeins}

\usepackage[normalem]{ulem}

\newcommand{\entropycolor}{orange}
\newcommand{\trainedcolor}{blue!70!cyan} %
\newcommand{\initialcolor}{red!80!white} %

\newtheoremstyle{Definition}%
{}%
{}%
{}%
{}%
{\upshape}%
{:}%
{}%
{}%
\newtheorem{definition}{Definition}
\newtheorem{example}{Example}

\makeatletter
\newsavebox{\measure@tikzpicture}
\NewEnviron{scaletikzpicturetowidth}[1]{%
  \def\tikz@width{#1}%
  \begin{lrbox}{\measure@tikzpicture}%
  \BODY
  \end{lrbox}%
  \pgfmathparse{#1/\wd\measure@tikzpicture}%
  \BODY
}
\makeatother

\usepackage[utf8x]{inputenc}

\lstdefinestyle{mystyle}{
  basicstyle=\footnotesize\ttfamily,
}
\lstdefinestyle{rewrite}{
  basicstyle=\scriptsize\sffamily,
  gobble=4,
}
\newcommand{\lstframeset}{1pt}
\lstdefinelanguage{Julia}%
  {morekeywords={abstract,break,case,catch,const,continue,do,else,elseif,%
      end,export,false,for,function,immutable,import,importall,if,in,%
      macro,module,otherwise,quote,return,switch,true,try,type,typealias,%
      using,while},%
   sensitive=true,%
   alsoother={\$},%
   morecomment=[l]\#,%
   morecomment=[n]{\#=}{=\#},%
   morestring=[s]{"}{"},%
   morestring=[m]{'}{'},%
}[keywords,comments,strings]%
\definecolor{vlightgray}{rgb}{0.95, 0.95, 0.95}

\definecolor{codebg}{rgb}{0.97, 0.97, 0.97}
 \sethlcolor{yellow}
\definecolor{codehl}{rgb}{0.75,0.87,0.95}
\sethlcolor{codehl}

\definecolor{codehl2}{rgb}{1.0,0.87,0.6}
\newcommand{\hlB}[1]{\sethlcolor{codehl2}\hl{#1}\sethlcolor{codehl}}
\newcommand{\hlBcolor}{rgb,255:red,229;green,166;blue,0}
\newcommand{\hlBtextcolor}{0.85,0.5,0.0}

\definecolor{codehl3}{rgb}{0.9,0.83,1.0}
\newcommand{\hlC}[1]{\sethlcolor{codehl3}\hl{#1}\sethlcolor{codehl}}
\newcommand{\hlCcolor}{rgb,255:red,179;green,77;blue,153}
\newcommand{\hlCtextcolor}{0.7,0.3,0.6}

\newcommand{\M}[1]{\ensuremath{#1}}

\newcommand{\N}{\M{\mathbb{N}}}

\newcommand{\uop}[2]{\operatorname{#1}\ #2}

\newcommand{\kind}[1]{\textit{#1}}
\newcommand{\defkind}[4]{
  \kind{#1} \qquad & #2 & #3\quad & #4 \\
}
\newcommand{\defkindnobr}[4]{
  \kind{#1} \qquad & #2 & #3\quad & #4
}
\newcommand{\assign}{::=}
\newcommand{\variantor}{\mid}

\newcommand{\etna}{\textsc{Etna}\xspace}
\newcommand{\rocq}{\textsc{Rocq}\xspace}

\newcommand{\letin}[2]{\mathrm{let\ }#1 = #2\mathrm{\ in}}

\newcommand{\dif}[3]{\mathrm{if\ }#1\mathrm{\ then\ }#2\mathrm{\ else\ }#3}

\newcommand{\dice}{\textsc{Dice}\xspace}
\newcommand{\ld}{\textsc{Loaded Dice}\xspace}

\newcommand{\Ourframework}{Our approach\xspace}
\newcommand{\ourframework}{our approach\xspace}

\newcommand{\quickcheck}{QuickCheck\xspace}
\newcommand{\quickchick}{QuickChick\xspace}

\newcommand{\SpecEntropy}{Specification Entropy\xspace}

\newcommand{\specentropy}{specification entropy\xspace}

\newcommand{\mono}[1]{\texttt{#1}}

\newcommand{\mathdefeq}{\vcentcolon=}

\newcommand{\defeq}{\triangleq}

\NewEnviron{irules}{\begin{gather*}\BODY\unpenalty\unskip\unpenalty\end{gather*}}

\RenewEnviron{irules}{\BODY}

\usepackage{tikz, tikz-cd}
\usetikzlibrary{patterns,calc,backgrounds, shapes.geometric}
\usetikzlibrary{decorations.markings}

\tikzstyle{nnf}=[
>=stealth,font=\small,auto,scale=0.7,every node/.style={scale=0.7}
]
\tikzstyle{extnode}=[
draw,circle,inner sep=2pt,fill=white
]

\tikzstyle{bddroot}=[
draw,fill=white,inner sep=2.5pt, font=\small
]

\tikzstyle{leafnode}=[
draw,fill=gray!20,inner sep=3.5pt
]
\tikzstyle{constnode}=[
draw,fill=white,inner sep=3.5pt
]
\tikzstyle{label}=[
fill=white,inner sep=2.5pt
]

\tikzstyle{acarrow}=[
decoration={markings,mark=at position 1 with {\arrow[scale=0.6]{>}}},
postaction={decorate},
shorten >=0.4pt,
>=latex,
line width=0.1
]

\tikzstyle{bnarrow}=[
decoration={markings,mark=at position 1 with {\arrow[scale=1.5]{>}}},
postaction={decorate},
shorten >=0.7pt,
>=latex,
line width=0.3
]
\tikzstyle{bayesnet}=[
>=latex, thick, auto
]
\tikzstyle{bnnode}=[
draw,ellipse,minimum size=7mm,inner sep=1pt,font=\small
]
\tikzstyle{cpt}=[
font=\footnotesize
]

\tikzstyle{graph}=[
>=stealth,font=\small,auto,scale=1,every node/.style={scale=1}
]
\tikzstyle{node}=[
draw,circle,inner sep=3pt,fill=white
]

\tikzstyle{bdd}=[
>=latex, thick, >=stealth, font=\small,auto,scale=0.9,every node/.style={scale=0.9}
]
\tikzstyle{bddnode}=[
draw,circle,inner sep=0pt,fill=white,minimum size=5.5mm
]

\tikzstyle{highedge}=[
line width=0.9
]
\tikzstyle{lowedge}=[
line width=0.9,dotted
]
\tikzstyle{bddterminal}=[
draw,fill=gray!20,inner sep=2.5pt, font=\small
]



\newcommand{\currentsidemargin}{%
  \ifodd\value{page}%
    \oddsidemargin%
  \else%
    \evensidemargin%
  \fi%
}

\newlength{\whatsleft}

\newcommand{\eightgraph}[9]{
    \begin{subfigure}[b]{#8}
        \begin{tikzpicture}
            \begin{axis}[
            height=2cm,width=\textwidth-18pt,
            grid=major,
            xlabel near ticks,
            ylabel near ticks,
            scale only axis,
            xticklabel style={
                /pgf/number format/fixed,
                /pgf/number format/1000 sep={,\!},
                /pgf/number format/precision=0,
                font=\tiny
            },
            xlabel style = {font=\scriptsize},
            ylabel style = {font=\scriptsize},
            title style = {font=\small},
            title={#1 #2},
            xlabel={#6},
            ylabel={},
            xticklabel style = {font=\tiny,yshift=2pt},
            yticklabel style = {font=\tiny,xshift=2pt},
            xminorticks=false,
            yminorticks=false,
            scaled x ticks=false,
            ybar,
            ymin=0,
            bar width=#7,
            enlarge x limits=#9,
            xtick=data,
                ]

                \addplot[ybar, bar shift=-#7-0.8, fill=\initialcolor!70, draw=\initialcolor!80, area legend] table[x={val}, y={probability}] {#3}; %
                \addplot[ybar, fill=\trainedcolor!70, draw=\trainedcolor!80, area legend] table [x={val}, y={probability}] {#4}; %
                \addplot[ybar, bar shift=#7+0.8, fill=black!30, draw=black, area legend, pattern=north east lines] table [x={val}, y={probability}] {#5}; %

            \end{axis}
        \end{tikzpicture}
    \end{subfigure}
}

\usepackage{epigraph}
\AtBeginDocument{}
\begin{document}
	
\title{Tuning Random Generators}
\subtitle{Property-Based Testing as Probabilistic Programming}

\author{Ryan Tjoa}
\orcid{0009-0003-0731-5398}
\affiliation{%
	\institution{University of Washington}
	\city{Seattle}
	\country{USA}
}
\email{rtjoa@cs.washington.edu}

\author{Poorva Garg}
\orcid{0000-0003-4753-3974}
\affiliation{%
	\institution{University of California, Los Angeles}
	\country{USA}
}
\email{poorvagarg@cs.ucla.edu}

\author{Harrison Goldstein}
\orcid{0000-0001-9631-1169}
\affiliation{%
	\institution{University of Maryland}
	\city{College Park}
	\country{USA}
}
\email{me@harrisongoldste.in}

\author{Todd Millstein}
\orcid{0000-0002-2031-1514}
\affiliation{%
	\institution{University of California, Los Angeles}
	\country{USA}
}
\email{todd@cs.ucla.edu}

\author{Benjamin C. Pierce}
\orcid{0000-0001-7839-1636}
\affiliation{%
	\institution{University of Pennsylvania}
	\city{Philadelphia}
	\country{USA}
}
\email{bcpierce@seas.upenn.edu}

\author{Guy Van den Broeck}
\orcid{0000-0003-3434-2503}
\affiliation{%
	\institution{University of California, Los Angeles}
	\country{USA}
}
\email{guyvdb@cs.ucla.edu}

\begin{abstract}
  \input{abstract}

\end{abstract}

\ccsdesc[300]{Software and its engineering~Software testing and debugging}
\ccsdesc[300]{Mathematics of computing~Probabilistic inference problems}
\ccsdesc[300]{Mathematics of computing~Automatic differentiation}

\maketitle

\section{Introduction}

Property-based testing (PBT) is a powerful~\cite{ artsTestingAUTOSARSoftware2015,artsTestingTelecomsSoftware2006,hughesMysteriesDropBoxPropertyBased2016}
and widely-studied~\cite{ hughesExperiencesQuickCheckTesting2016,hughesHowSpecifyIt2019,goldsteinParsingRandomness2022}
software testing technique that validates a system under test with
respect to an executable specification by evaluating it on many
randomly generated inputs. For example, when testing a
sorting function, a user might
write the property
{\small
\[ \forall \mono{l. } \mono{isSorted (sort l)}, \]}
which specifies that the result of sorting a list should be sorted.
To check this property, the PBT framework generates hundreds or
thousands of inputs 
to the property (lists \mono{l}) and checks the statement
with respect to each one. Since testing performance is entirely dependent on
the distribution of test inputs, a great deal of PBT research has focused
on how to quickly 
generate inputs that find more bugs, faster~\cite{claessenGeneratingConstrainedRandom2015,goldsteinParsingRandomness2022}.

The standard way that PBT users generate test inputs is via
{\em generators} --- programs that describe how to sample
test inputs from some random distribution. There has been extensive
research both on domain-specific languages for manually constructing
generators~\cite{Luck} and on methods for automatically deriving generators from data type
definitions~\cite{generic-random} or inductive
relations~\cite{leo-good-generators}.

But despite all this flexibility and automation, a key challenge remains:
\emph{generator tuning}. In order to achieve a good distribution of test
inputs---one where ``interesting'' inputs appear often---the
programmer has to manually decide on the weights of the individual 
random choices that are made as the generator executes.

For example, the
following generator is the one a developer might try writing down when testing
the above \mono{sort} property.
We write it in our \ld\footnote{
\ld is available at \url{https://github.com/Tractables/Alea.jl/tree/loaded-dice}.
} probabilistic programming language (\autoref{sec:loaded-dice}), which is an embedded domain-specific language in Julia;
\mono{@match} is a Julia macro defined by \ld to implement pattern matching.

\begin{lstlisting}
genList(sz) = # generates lists up to length sz
  @match sz (
    0 X$\to$X Nil(),
    X$\label{line:one-of}$XS(sz') X$\to$X oneOf [ Nil(), Cons(genNat(), genList(sz')) ])
\end{lstlisting}
If a developer inspects the test cases produced by this generator,
they will quickly
notice an obvious problem --- \(50\%\) of the generated test cases will be
empty lists! This is because the \texttt{oneOf} combinator in the generator makes a uniform random choice between
the two constructors for the data type.  Half of the time, it chooses
\texttt{Nil}, and the other half of the time, it chooses
\texttt{Cons}. Note that this means not only are half the
  lists empty, but half of the rest are length 1, and so on.

Even if it is obvious to the developer that this distribution is a poor choice
for testing,
it may not be obvious how to improve the situation. One could use
\texttt{freq}, a replacement for \texttt{oneOf} that allows the user to manually
add weights to the random choice. For example, the following use of \texttt{freq} chooses \mono{Cons} two-thirds of the time.
{\small
\[\mono{freq [1 $\Rightarrow$ Nil(), 2 $\Rightarrow$ Cons(genNat(), genList(sz'))]}\]}

But which weights should the programmer use?
In general, and especially as generators get more complicated, it can be
a significant challenge to understand how changing weights changes the
final distribution.
With the above, the distribution of list lengths is roughly\footnote{It differs from the geometric distribution because it is truncated at the initial \mono{sz}, as \mono{Nil} is always chosen when \mono{sz} is 0.} the geometric distribution with success probability two-thirds.
Reasoning about the distributions of more complex generators, and how they depend on the weights, quickly increases in difficulty.
Indeed, recent work on PBT
usability~\cite{tyche} cited tuning as a source
of ``mental strain'' for developers who felt like they needed to ``study
probability and statistics'' to understand how to tune a
generator to suit their needs.

In this paper, we address this issue by providing developers with techniques for
{\em automatically} tuning generators. Concretely,
users can write down generators
with {\em symbolic weights} that are not yet determined, then specify an
{\em objective function} that the weights should attempt to optimize.
We present an offline approach to automatically learn values for these
weights to optimize for a given objective function.

\Ourframework is flexible enough to handle a wide variety of objective
functions.
If the developer has an intuition about the distribution they want (e.g.,
that the distribution of lengths of generated lists should be uniform)
they can simply
        optimize the generator to try to match that distribution.
If they don't know the precise distribution they're after, they can
instead favor diversity of test cases by
optimizing for {\em entropy}~\cite{shannon}, a standard metric for diversity.
Finally, if the developer has a notion of ``validity'' that they want to
        maintain, they can combine entropy with adherence to some
        specification of validity.

\Ourframework automatically tunes PBT generators to optimize
objective functions by expressing the generators as programs in an extension of
\dice~\cite{dice}, a discrete probabilistic programming language (PPL).
PPLs and generator
languages both deal with randomness, but, for our purposes, \dice has a significant advantage:
it can perform exact probabilistic inference, computing a representation of
the full distribution of a given generator.
Furthermore, the inference strategy in \dice is
differentiable, so we can use
gradient descent to optimize symbolic generator weights for a given
objective. We design and implement \ld, an extension of \dice that supports
differentiation and parameter learning, and use it as a language for
generators.

To make tuning feasible, we address two key performance challenges.
First, probabilistic inference for discrete PPLs is \#P-hard
in general~\cite{probwmc},
which in turn makes computing gradients \#P-hard as well. We address this problem by
choosing our PPL carefully: \dice compiles
to binary decision diagrams (BDDs) that
naturally exploit program structure to scale probabilistic inference.
Since the
computation of gradients in \ld happens on the same BDDs, it can leverage the
same scaling benefits.
The second performance challenge arises from the fact that the na\"ive way to
compute our proposed objective functions requires enumerating the whole space of
possible test cases. This is infeasible, so we adapt a scalable gradient
estimation technique, REINFORCE~\cite{REINFORCE}, to our context of generator tuning.
At a high level, REINFORCE allows us to avoid this large enumeration
by replacing it with sampling.

With these elements in place, we present multiple examples that demonstrate the
effectiveness of our approach in steering the distributions of the generators.
We also perform a thorough evaluation on PBT benchmarks, demonstrating
that, when tuned for diversity and validity, the generators lead to a
3.1--7.4$\times$ speedup in bug finding.

Following a high-level overview in
\autoref{sec:overview}, we offer the following contributions:
\begin{itemize}
  \item We describe a space of generator-independent objective functions that can target a specific distribution or increase the diversity and validity of the generator (\autoref{sec:objectives}).
  \item We describe the design and implementation of \ld, a PPL that allows weights to be learned for these objectives by extending \dice with automatic differentiation (\autoref{sec:loaded-dice}).
  \item We show techniques for constructing generators more amenable to tuning, and how to derive such generators from inductive type definitions (\autoref{sec:constructing-tunable-generators}).
  \item We present training techniques to achieve these objectives in practice. In particular, we adapt REINFORCE, a gradient estimation technique, to the context of generator tuning, in order to efficiently optimize for entropy-based objectives (\autoref{sec:training-techniques}).
  \item To evaluate our approach, we use \ld to tune a diverse collection of type-based generators for validity and diversity and to tune a handwritten STLC generator for a particular distribution, improving the speed at which they find bugs on existing benchmarks (\autoref{sec:etna}).
\end{itemize}

\section{Overview} \label{sec:overview}

In this section, we overview \ourframework, focusing in particular on
how it can benefit PBT users.
\vspace{-1pt}

\subsection{The Basics of Generator Tuning}\label{sec:overview-basics}

\begin{figure}[b]
\begin{lstlisting}
@type Color = R() | B()
@type Tree = Leaf() | Branch(Color, Tree, Nat, Nat, Tree)

genColor() = freq([X$\theta_\mono{red}$X X$\Rightarrow$X R(), 1 - X$\theta_\mono{red}$X X$\Rightarrow$X B()])

genTree(size) =
  @match size (
    O       X$\to$X Leaf(),
    S(n) X$\to$X (
      w = @match size ( 1 X$\Rightarrow$X X$\theta_1$X, 2 X$\Rightarrow$X X$\theta_2$X, 3 X$\Rightarrow$X X$\theta_3$X, 4 X$\Rightarrow$X X$\theta_4$X, 5 X$\Rightarrow$X X$\theta_5$X );
      freq([w X$\Rightarrow$X Leaf(),
            1-w X$\Rightarrow$X  Branch(genColor(), genTree(n), genNat(), genNat(), genTree(n))])))

G = genTree(5)
\end{lstlisting}
\vspace{-4pt}
\caption{A generator of (not necessarily valid) red-black trees, using
symbolic weights.
The macros \mono{@type} and \mono{@match}, provided by the \ld embedding in Julia, implement algebraic datatypes and pattern matching.
The match expression computing \mono{w} allows the weights to depend on \mono{size}, as described in \autoref{sec:splitting}.
}\label{fig:gen-tree}
\end{figure}

Since the beginning of PBT, in
QuickCheck~\cite{claessenQuickCheckLightweightTool2000}, generators have been a
core part of the PBT process. PBT frameworks provide a domain-specific language
(DSL) for expressing and combining generators, and programmers can use that language
to design arbitrarily complicated distributions of test inputs for their programs.

While the power of generator DSLs is generally a significant benefit, the
complexity of PBT generators leads to some important challenges. The key
challenge we focus on in this paper is {\em tuning}, which is the process of
choosing the weights with which different random choices in the generator are
made.

To illustrate tuning and demonstrate why it is difficult,
we start with a toy example. 
Consider this
\ld generator of characters from \mono{`a'}--\mono{`e'}:
\enlargethispage{2pt}
\begin{lstlisting}[breaklines=true]
G = freq([
  X$\theta_1$X X$\Rightarrow$X freq([X$\theta_2$X X$\Rightarrow$X 'a', X$\theta_3$X X$\Rightarrow$X 'b', X$\theta_4$X X$\Rightarrow$X 'c']),
  X$\theta_5$X X$\Rightarrow$X freq([X$\theta_6$X X$\Rightarrow$X 'c', X$\theta_7$X X$\Rightarrow$X 'd', X$\theta_8$X X$\Rightarrow$X 'e'])])
\end{lstlisting}
The generator \mono{G} uses the \mono{freq} combinator to make
weighted random choices between different options; each $\theta$ is a
placeholder for a number that decides the relative weight of that particular
choice. For example, if $\theta_2$, $\theta_3$, and $\theta_4$ were all $1$, the
\mono{freq} containing them would make a uniform choice. If $\theta_2$ were
changed to $2$, the value \mono{`a'} would be chosen twice as often as
\mono{`b'} or \mono{`c'}.

Now, suppose we want to ensure that \mono{G} has a uniform distribution over the
five characters --- that is, that each is sampled \(20\%\) of the time. We
encourage readers to take a moment to try to come up with values for
$\theta_1$--$\theta_8$ that produce the appropriate distribution. Even for this
very simple example, it is not totally obvious!

\Ourframework entirely automates this reasoning. The user can simply write a
target distribution and then ask that the generator be trained to match that
distribution:
\begin{lstlisting}
target = ['a' X$\Rightarrow$X 0.2, 'b' X$\Rightarrow$X 0.2, 'c' X$\Rightarrow$X 0.2, 'd' X$\Rightarrow$X 0.2, 'e' X$\Rightarrow$X 0.2]
objective = -kl_divergence(target, G)
\end{lstlisting}
We discuss KL divergence later in detail; for now, read line 2 as ``pick weights
in \mono{G} such that the final distribution is as close as possible to
\mono{target}.'' Given these inputs, our approach automatically learns weights proportional to
$\{\theta_1\!\mapsto\!1,\; \theta_2\!\mapsto\!2,\; \theta_3\!\mapsto\!2,\; \theta_4\!\mapsto\!1,\; \theta_5\!\mapsto\!1,\; \theta_6\!\mapsto\!1,\; \theta_7\!\mapsto\!2,\; \theta_8\!\mapsto\!2\}$,
which achieves the desired distribution.

\subsection{Approximating Known Distributions}\label{sec:overview-approx}

Next, we move on to a more realistic example, continuing to explore how
\ourframework allows developers to tune generators according to a concrete
desired distribution.

Consider the generator in \autoref{fig:gen-tree}, which generates
random color-labeled binary trees (some subset of these trees will be valid
red-black trees~\cite{red-black-trees}, but this generator does not
ensure that invariant). The function \mono{genTree} takes a maximum tree
\mono{size} and then generates trees up to that size.  If \texttt{size} is
non-zero, \mono{genTree} makes a random choice: with probability
\mono{w} it generates a leaf, and with probability \mono{1 - w} it generates a color, key, and value for an internal
node, then recurses to generate the two child trees with reduced \texttt{size}.
If \texttt{size} is zero, then it always produces a leaf.

Starting from this generator, a developer might have some ideas for what they
would like the distribution of trees to look like. For example, they may tune
the generator to produce relatively few very small trees (height 1 and 2) and
proportionally more large trees (height 4 and 5).
They can specify their desired distribution over heights by using a \mono{height} function and \mono{kl_divergence}:
\begin{lstlisting}
target = [1 X$\Rightarrow$X 0.1, 2 X$\Rightarrow$X 0.1, 3 X$\Rightarrow$X 0.2, 4 X$\Rightarrow$X 0.3, 5 X$\Rightarrow$X 0.3]
objective = -kl_divergence(target, height(G))
\end{lstlisting}
This is similar to the
example above, but now both the generator and the objective are significantly
more realistic: the generator is modeled after one that appears often in the PBT
literature, and the objective adheres to a common PBT principle that larger test
inputs often find more bugs.\footnote{As observed by \citet{etna} and shown in \autoref{sec:etna}, it can also be beneficial to tune for \textit{smaller} inputs.}

\begin{figure}[!t]
        \centering
	\begin{subfigure}[t]{0.9\linewidth}
		\centering
            \ifcharts
		\begin{tikzpicture} 
			\begin{axis}[%
				hide axis,
				xmin=10,
				xmax=50,
				ymin=0,
				ymax=0.4,
				legend columns=-1,
				legend style={draw=none,legend cell align=left,column sep=1ex}
				]
				\addlegendimage{ybar, fill=\initialcolor!70, draw=\initialcolor!80, area legend}
				\addlegendentry{\small Initial};
				\addlegendimage{ybar, fill=\trainedcolor!70, draw=\trainedcolor!80, area legend}
				\addlegendentry{\small Tuned};
				\addlegendimage{ybar, fill=black!30, draw=black, pattern=north east lines, area legend}
				\addlegendentry{\small Objective};
			\end{axis}
		\end{tikzpicture}
            \fi
	\end{subfigure}

	\newcommand{\treewidth}{0.23\textwidth}
	\newcommand{\treebarwidth}{2.5pt}
	\newcommand{\treeenlargex}{0.15}
	\newcommand{\nonbespokewidth}{0.24\textwidth}
	\newcommand{\nonbespokebarwidth}{2.5pt}
	\newcommand{\nonbespokeenlargex}{0.15}
	\newcommand{\bespokewidth}{0.4\textwidth}
	\newcommand{\bespokebarwidth}{2.3pt}
	\newcommand{\bespokeenlargex}{0.06}
        \begin{minipage}{\textwidth}
        \ifcharts
            \centering
            \raisebox{0.76cm}{%
                \begin{tikzpicture}
    \node[rotate=90, align=center] at (0,0) {\small Uniform\\\scriptsize Probability};
                \end{tikzpicture}%
            }%
			\hspace{-8pt}
            \eightgraph{RBT Type-Based}{}{\bsttbinitdepth}{\rbttbunifdepth}{\treeuniftarget}{}{\treebarwidth}{\treewidth}{\treeenlargex}%
			\hspace{-10pt}
            \eightgraph{STLC Type-Based}{}{\stlctbinitdepth}{\stlctbunifdepth}{\stlctbuniftarget}{}{\nonbespokebarwidth}{\nonbespokewidth}{\nonbespokeenlargex}%
			\hspace{-10pt}
            \eightgraph{STLC Bespoke}{}{\stlcbespokeinitdepth}{\stlcbespokeunifdepth}{\stlcuniftarget}{}{\bespokebarwidth}{\bespokewidth}{\bespokeenlargex}%
        \fi
        \end{minipage}
        \begin{minipage}{\textwidth}
        \ifcharts
            \centering
            \raisebox{1.24cm}{%
                \begin{tikzpicture}
    \node[rotate=90, align=center] at (0,0) {\small Linear\\\scriptsize Probability};
                \end{tikzpicture}%
            }%
			\hspace{-8pt}
            \eightgraph{}{}{\bsttbinitdepth}{\rbttblindepth}{\treelineartarget}{Tree Height}{\treebarwidth}{\treewidth}{\treeenlargex}%
			\hspace{-10pt}
            \eightgraph{}{}{\stlctbinitdepth}{\stlctblindepth}{\stlctblineartarget}{AST Height}{\nonbespokebarwidth}{\nonbespokewidth}{\nonbespokeenlargex}%
			\hspace{-10pt}
            \eightgraph{}{}{\stlcbespokeinitdepth}{\stlcbespokelindepth}{\stlclineartarget}{AST Height}{\bespokebarwidth}{\bespokewidth}{\bespokeenlargex}%
        \fi
        \end{minipage}
    \vspace{-5pt}
		\pgfplotsset{every axis/.style={}}
        \caption{Tuning the distributions of heights of generated values to be uniform or linear for several generators. The top row trains the generators to have a uniform distribution over heights, and the bottom row trains the generators to have a linear distribution over heights.\protect\footnotemark}
        \label{fig:height-tunings}
    \vspace{-5pt}
\end{figure}

\footnotetext{
The tuned distributions of the AST heights in the STLC bespoke generator do not  
exactly match the target distribution, but they are closer to their objectives:
tuning improves KL divergence from 0.44 to 0.22 for the uniform target distribution, and from 0.92 to 0.27 for the linear target distribution.}

\autoref{fig:height-tunings} demonstrates this process on a wider range of
examples. 
We show the results of tuning three generators: the above generator for
color-labeled binary trees, a generator for terms in the simply-typed lambda calculus
(STLC) that may or may not be well-typed, and a ``bespoke'' handwritten generator for well-typed STLC terms (generator shown in Appendix B).
For each, we tune the distributions to have either
a uniform or a linear relationship between data structure height and sampling
frequency.
The charts show that the tuned generators (in blue) match the
target distribution (in gray) much better than the untuned\footnote{Throughout this paper, an ``untuned'' generator is one in which each random choice is uniformly distributed.} versions (in red).

\subsection{Under the Hood}

Sections \ref{sec:objectives}, \ref{sec:loaded-dice}, and
\ref{sec:training-techniques} discuss the technical details of \ourframework and
its implementation in significant detail; here we simply give a high level
picture.

Our key observation is that by implementing generators in a probabilistic programming language, we get easy access to algorithms that can be used to automate tuning. In particular, we choose \dice~\cite{dice} as our starting point because it
provides scalable procedures for {\em exact probabilistic inference} ---
computing a closed-form representation of the whole generator distribution.
This means that if we implement a generator in \dice, we can compute precisely
how well that generator matches a distribution requested by the user.

\begin{figure}[!t]
	\centering
	\begin{tikzpicture}[
		node distance=3.8cm,
		auto,
		thick,
		box/.style={rectangle, draw, rounded corners=3pt, minimum width=1.5cm, minimum height=0.6cm, align=center, font=\footnotesize},
		arrow/.style={->, >=stealth, thick}
	]
		\node[box] (tuned) {Tuned\\Weights};
		\node[box, left of=tuned] (prob) {Symbolic\\Loss};
		\node[box, xshift=-0.2cm, left of=prob,yshift=0.5cm] (bdd) {BDD};
		\node[box, xshift=0.2cm,left of=bdd] (ld) {\ld};
		\node[box, xshift=-0.2cm, left of=prob, yshift=-0.5cm] (objective) {Probabilistic Objective};
		
		\draw[arrow] (ld) -- node[above] {\footnotesize Compilation} (bdd);
		\draw[arrow] (bdd) -- node[above] {\footnotesize\;\;Model Counting} (prob);
		\draw[arrow] (prob) -- node[above] {\footnotesize Gradient} (tuned);
		\draw[arrow] (prob) -- node[below] {\footnotesize Descent} (tuned);
		\draw[arrow] (objective) -- (prob);

	\end{tikzpicture}
	\caption{Tuning Generators in \ld.}
  \label{fig:teaser}
\end{figure}

We extend \dice to a probabilistic programming system called \ld with two significant additions.
First, in \dice all weights must be concrete numbers; \ld instead has {\em symbolic weights} (the $\theta$s in the generators above),
which allow the programmer to choose points in the generator where weights should
be learned automatically. Second, \ld extends \dice's inference algorithm to
compute {\em gradients}; we can compute not only the distribution of a
generator, but also how the symbolic weights should be changed to
maximize a given objective function.  

With these features in place, \autoref{fig:teaser} shows the workflow of our approach. A generator in \ld with symbolic weights is first compiled to a binary decision diagram (BDD) that represents all
possible executions of the program. The benefit of this representation is that we can then perform exact probabilistic inference on the program as a linear pass over the BDD~\cite{dice}. Specifically, 
the procedure of \emph{weighted model counting} on the BDD generates expressions over the symbolic weights that correspond to the probability distribution of the generator.
These expressions are then plugged into the user's desired objective function, and the result is a differentiable expression over the symbolic weights in the generator that can be optimized using gradient descent. The system repeatedly
computes the gradient of the generator with respect to the objective function and
then nudges the generator weights in the appropriate direction. In the end, the
training process is likely to stabilize on a choice of weights that gets close to the desired distribution~\cite{BlumHopcroftKannan2020}.

\subsection{More General Tuning Objectives}\label{sec:overview-general-objectives}

Matching developer-specified distributions is certainly useful, but in
real-world settings a developer might not actually know what distribution will
be best for testing. To address this situation, we have identified two 
generally useful properties of a distribution of test cases that we can use as
generic tuning objectives.

One natural objective that a developer may want for their generator is
maximizing the diversity or {\em entropy} of the generated values. For
generators producing unconstrained values, this can be an ideal shortcut to a
balanced distribution. For example, maximizing the entropy of the
\mono{`a'}--\mono{`e'} generator from \autoref{sec:overview-basics} gives the same
weights as the more explicitly defined uniform distribution.

But often, values are constrained by a {\em validity predicate} --- for example,
the color-labeled tree example from \autoref{sec:overview-approx} may be
expected to produce valid red-black trees for the purposes of testing functions
that require red-black tree validity as an invariant. For these situations,
maximizing entropy alone is not sufficient. We therefore define a notion of {\em
specification entropy}, which attempts to simultaneously optimize entropy and
validity. \autoref{fig:stlc-unique-types} shows the results
of tuning an STLC generator to increase the entropy
of the types of generated terms,
and simultaneously increase the
likelihood of well-typedness.
\autoref{fig:stlc-unique-dist-tuning} visualizes
how the distribution over types changes over time as we
tune the weights;
\autoref{fig:stlc-unique-types-cumulative} contrasts the initial generator with the tuned version by comparing the number of unique types each generates as the number of samples increases.
The trained version produces terms with far more
diverse and interesting types.
At the beginning of tuning, ill-typed terms and terms of type bool
comprise 66\% and 25\% of samples, respectively. After tuning,
neither any single type nor ill-typed terms comprise more than 0.5\% of samples.

\begin{figure}[!t]
    \begin{subfigure}[b]{0.55\textwidth}
      \ifcharts
    \hspace{-15pt} 
          \includegraphics[width=0.96\textwidth]{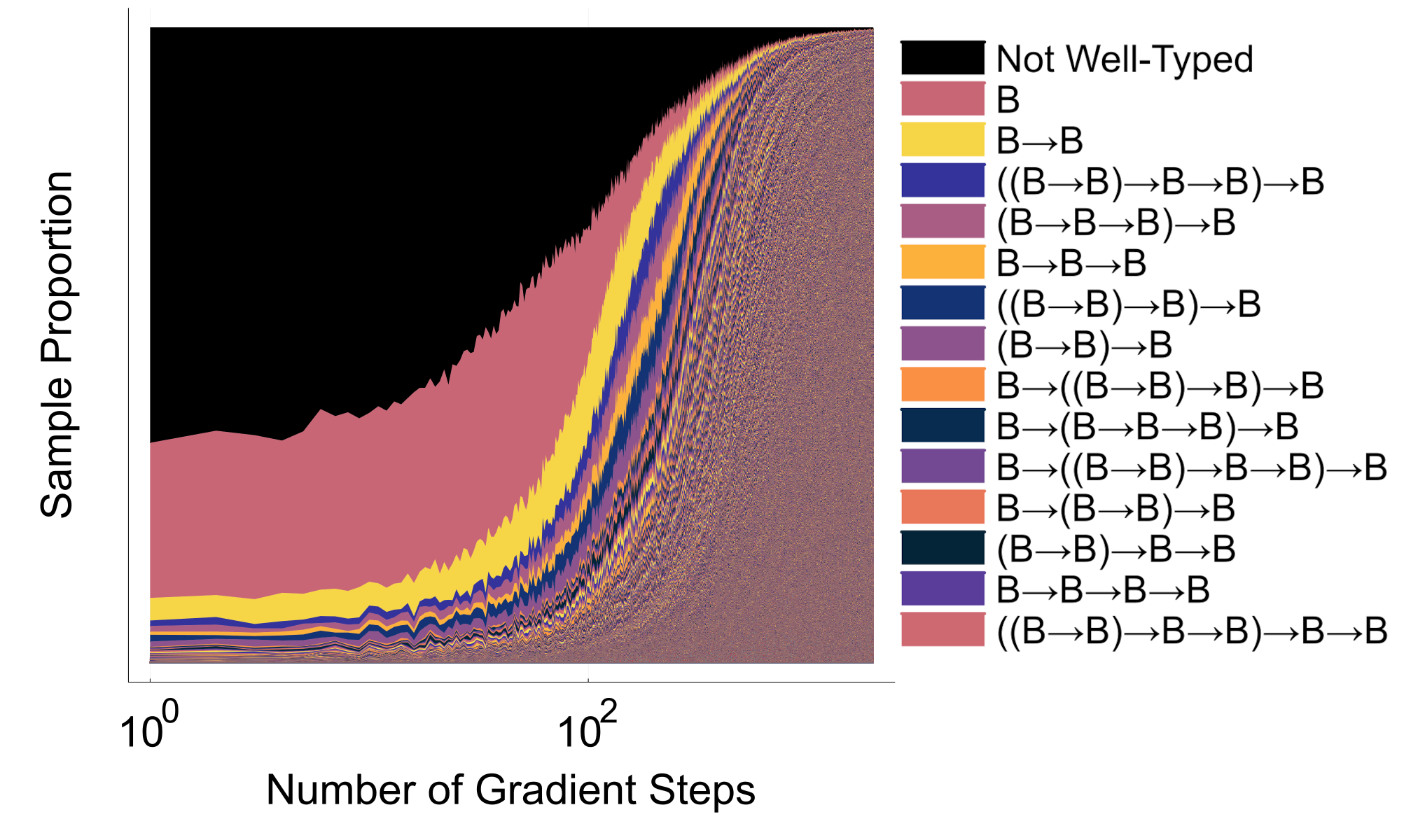}
      \fi
      \vspace{-1pt}
      \caption{
  A visualization of 
  how the distribution over\\types changes as we
  tune the weights. The x-axis\\uses log scaling.
    }
    \label{fig:stlc-unique-dist-tuning}
    \end{subfigure}%
    \hspace{-10pt}
    \begin{subfigure}[b]{0.39\textwidth}
    \hspace{-5pt}
              \ifcharts
          \begin{tikzpicture}
            \begin{axis}[
                  height=3.2cm,width=\textwidth-45pt,
                  grid=major,
                  xlabel near ticks,
                  ylabel near ticks,
                  scale only axis,
                  xtick={10, 1000, 100000},
                  xticklabel style={font=\scriptsize},
                  ytick={10,100,350,1000,9498},
                  xmode=log,
                  ymode=log,
                  xlabel style = {font=\footnotesize},
                  ylabel style = {font=\footnotesize},
                  xlabel={Number of Samples},
                  ylabel={\# of Unique Types},
                  yticklabel style = {font=\scriptsize},
                  xminorticks=false,
                  yminorticks=false,
                  legend style={at={(0.98,0.02)}, anchor=south east, font=\scriptsize},
                  legend entries={Initial,Tuned}
              ]
              \addplot[mark=none, thick, \initialcolor, mark size = 1.0pt] table[x index={0}, y index={1}] {\cumulativeuniqdata}; %
              \addplot[mark=none, thick, \trainedcolor, mark size = 1.0pt] table [x index={0}, y index={2}] {\cumulativeuniqdata}; %
            \end{axis}
          \end{tikzpicture}
              \fi
      \caption{Cumulative unique types throughout sampling, before and after tuning.
      Both axes use log scaling.
    }
      \label{fig:stlc-unique-types-cumulative}
    \end{subfigure}
    \caption{Results for tuning an STLC generator for unique types using \specentropy.
  For brevity, the legend of (a) shows only the most common types and abbreviates ``Bool'' as ``B.''
    }
    \label{fig:stlc-unique-types}
  \end{figure}

\section{Objective Functions} \label{sec:objectives}

Users of property-based testing sometimes have intuitions about what specific generator distributions are desirable for their use case~\cite{goldsteinPropertyBasedTestingPractice2024}.
This section describes how these intuitions can be expressed as objective functions for automated generator tuning in \ourframework. We first provide some preliminaries,
and then \autoref{sec:objectives-target} describes how a user can specify an objective function to tune for a target distribution. Then, \autoref{sec:objectives-spec-entropy} describes an objective function to improve the diversity and validity of their test cases.

\subsection{Preliminaries and Notation}\label{sec:objectives-setup}

The random choices made by a generator induce a particular distribution over
the test cases it can produce.
Let $G$ denote a generator for test cases of type $T$ with $n \in \N$ symbolic weights.
We represent an \textit{assignment} to its weights as \(\assignment \in [0, 1]^n\).
Then,
we denote the probability distribution induced by $G$ instantiated with those weights as \(p_{G, \assignment}\).

Now, for automated generator tuning, we need a measure of how good the generator distribution is. For this purpose, we define an \textit{objective function} as follows:

\begin{definition}[objective function]
	Given a generator \(G\) with \(n\) symbolic weights, an objective function \(f: [0, 1]^n \rightarrow \mathbb{R}\) is defined such that for two assignments of weights, \(w\) and \(w'\), if \(f(\assignment) > f(\assignment')\) then the user prefers distribution \(p_{G, \assignment}\) over \(p_{G, \assignment'}\).
\end{definition}

\begin{example}
\label{example:max-prob-b}
If $G$ generates characters \mono{`a'-`e'}, the objective function
$f(\assignment) = p_{G,\assignment}(\mono{`b'})$ simply maximizes the probability of generating \mono{`b'}.
\end{example}

We provide two useful families of objective functions, which we specify below. We introduce the target objective function to tune for a particular distribution and the specification entropy objective function to improve the diversity and validity of generated test cases.

\subsection{Objective Function to Target a Distribution}\label{sec:objectives-target}

As stated in \autoref{sec:overview}, PBT users sometimes desire a particular distribution over a feature of their generated test cases (say, one might want a generator for RBTs to produce a uniform distribution over tree heights).
In fact, there are numerous tools that record the generator distribution to make it easier for the user to visualize~\cite{stewartTestQuickCheckHackage2024,tyche, PierceSF4},
but the user still has to manually update and reason about the weights to adjust their distribution.
Through automated generator tuning, our approach optimizes weights such that the generator distribution approaches the user's desired distribution.

To define an objective function for this task, we first capture how a generator can induce a distribution over a feature of
its generated test cases through the following definition:

\begin{definition}[push-forward of generator distribution]
	The push-forward of a generator distribution \(p_{G, \assignment}\) over type $T$ through a function \(g: T \rightarrow T'\) is the probability distribution \(p_{G, \assignment, g}\) over type \(T'\) such that
	\[\forall t' \in T', \;\;\; p_{G, \assignment, g}(t') = \sum_{t \in T, g(t) = t'} p_{G, \assignment} (t).\]
\end{definition}

\begin{example}
	Let $G$ be a generator over lists such that 
	\(p_{G,\assignment}([1, 2]) = 0.7\) and \(p_{G,\assignment}([2, 3]) = 0.3\) and
	$g$ be the length function for lists. Then, by the above definition, \(p_{G,\assignment,g}(2) = 1\).
\end{example}

Now, the objective function that aims for a particular distribution should minimize the distance between the generator distribution and the target distribution. To capture this notion, we use KL divergence~\cite{kldivergence} as the measure of how much one probability distribution differs from the other and define the \textit{target objective function} as follows:

\begin{definition}[target objective function]
	Given a generator \(G\) for test cases of type $T$, a function \(g: T \rightarrow T'\), and a target distribution \(\tilde{p}\) over values of type \(T'\), the target objective function is defined as the negative KL divergence between the target distribution and the push forward of the generator distribution through \(g\).
	
	\begin{equation*}\mathit{Target}(\assignment) \mathdefeq - \mathit{KLD}(\tilde{p}, p_{G, \assignment, g}) = - \sum_{t' \in T'} \tilde{p}(t') \log{\frac{\tilde{p}(t')}{p_{G, \assignment, g}(t')}}\end{equation*}
\end{definition}

\begin{example}
	Let \(G\) be a generator for red-black trees with symbolic weights \(\assignment\). Let \(g\) be a function that takes as input a red-black tree and outputs its height. Let the user-specified target distribution \(\tilde{p}\) be defined over the tree height as \(\{2 \rightarrow 0.3, 4 \rightarrow 0.7\}\), then the target objective function would be
	\[\mathit{Target}(\assignment) = - 0.3 \log \frac{0.3}{p_{G, \assignment, g}(2)} - 0.7 \log \frac{0.7}{p_{G, \assignment, g}(4)}. \]
\end{example}

We demonstrate the effectiveness of the target objective function in \autoref{fig:height-tunings} where we tune three different generators for uniform and linear distributions over a feature of the test cases. In \autoref{sec:etna}, we use the target objective function to leverage insight from \citet{etna} that smaller STLC terms find bugs faster, resulting in improved bug-finding speed.

To summarize, in this section we have shown how to turn a user-specified distribution into an objective function that we can optimize for. This is useful when the user has intuition about what the generator distribution should look like.

\subsection{Objective Function to Improve Diversity and Validity}\label{sec:objectives-spec-entropy}

If a user does not have a target distribution in mind, they may instead
want to tune the
generator weights to improve the diversity of their test cases and increase the
number of valid ones. This has the potential to speed up testing by exercising a
wider variety of program configurations and reducing the time spent generating
invalid inputs. This section describes the objective functions to optimize for
these distributional properties and how one can combine them.

\subsubsection{Targeting Diverse Generations}

In this section we examine how to improve
the diversity of the test
cases produced by a generator. For this purpose, we define the entropy objective function using the
information-theoretic notion of entropy~\cite{shannon} of a probability
distribution.

\begin{definition}[entropy objective function]
	The entropy objective function for a generator \(G\) is defined as the entropy of its generator distribution \(p_{G, \assignment}\).
	\begin{equation*}
		\mathit{Entropy}(\assignment) \mathdefeq H(p_{G, \assignment})
= - \sum_{t \in T} p_{G, \assignment}(t) \log{p_{G, \assignment}(t)}
	\end{equation*}
\end{definition}

Note that a uniform distribution over all possible test cases has the maximum entropy, thus maximizing the entropy objective function takes the generator distribution closer to a uniform distribution over all possible generations.

\subsubsection{Targeting Valid Generations}

Many common PBT examples have {\em preconditions} that define the set of
valid inputs to the program under test.
For instance, to test a program that inserts
elements into a red-black tree,
it is only useful to generate trees that satisfy
the red-black tree invariant. To perform automated tuning for this purpose, we define the \textit{specification objective function}.\footnote{
	This definition of the specification objective function is
	in accordance with \citet{semantic-loss}.
}

\begin{definition}[specification objective function]
	Given a generator \(G\) for test cases of type $T$ and a validity condition \(\phi : T \to \{0, 1\}\), the specification objective function is defined as
	\begin{equation*}
		\mathit{Specification}(\assignment) \mathdefeq \log{p_{G, \assignment, \phi}(1)}
= \log \bigg( \sum_{
        \substack{t \in T \\ \phi(t) = 1}
    } p_{G,\assignment}(t) \bigg).
	\end{equation*}
\end{definition}

\noindent
Intuitively, the specification objective function
attempts to maximize the probability that a generated test case meets the validity condition $\phi$.

\subsubsection{Targeting Diverse, Valid Generations} The previous subsections discussed objective functions to target diversity and validity.
However, these two objectives inherently conflict.
Tuning for diversity incentivizes large terms, which are more likely to be diverse but less likely to be valid. Tuning for validity incentivizes trivially valid terms such as empty trees or lists.
As a result, the common technique of combining objectives by taking their weighted sum is not effective here.
To resolve this tension, we introduce the \textit{specification entropy objective function}, which
targets the entropy of the generator distribution \textit{within} the space of valid test cases.

\begin{definition}[specification entropy objective function]
	Given a generator \(G\) and a validity condition 
	\(\phi : T \to \{0, 1\}\),
	the specification entropy objective function is defined as
	\begin{equation*}
		\mathit{SpecificationEntropy}(\assignment) \mathdefeq - \sum_{\substack{t \in T \\ \phi(t) = 1}} p_{G, \assignment}(t) \log{p_{G, \assignment}(t)}. 
	\end{equation*}
\end{definition}

\noindent
This objective function
aims to generate diverse test cases, except it 
disregards test cases that are invalid.

\begin{figure}[t]
	\centering
				\ifcharts
		\begin{tikzpicture}
			\begin{axis}[
				height=2.2cm,width=3.5cm,
				grid=major,
				xlabel near ticks,
				ylabel near ticks,
				scale only axis,
				xticklabel style={
					/pgf/number format/fixed,
					/pgf/number format/1000 sep={,\!},
					/pgf/number format/precision=0,
					font=\scriptsize
				},
				xlabel style = {font=\footnotesize},
				ylabel style = {font=\footnotesize},
				title style = {font=\small},
				title={Cumulative Unique Valid RBTs},
				xlabel={Number of Samples},
				ylabel={\# of Unique Valid RBTs},
				yticklabel style = {font=\scriptsize},
				xminorticks=false,
				yminorticks=false,
				scaled x ticks=false,
				legend style={
					at={(1.1,0.5)},
					anchor=west,
					font=\scriptsize,
					draw=none
				},
				legend entries={Untuned,
					Validity,
					Entropy,
					Specification Entropy
				},
				]
				\addplot[mark=none, thick, \initialcolor] table[x=Samples, y=UntunedGenerator.v] {\rbtablationcumulativeuniqdata};
				\addplot[mark=none, thick, teal] table[x=Samples, y=ValidBoundGenerator.v] {\rbtablationcumulativeuniqdata};
				\addplot[mark=none, thick, \entropycolor] table[x=Samples, y=EntropyBoundGenerator.v] {\rbtablationcumulativeuniqdata};
				\addplot[mark=none, thick, \trainedcolor] table[x=Samples, y=SEBoundGenerator.v] {\rbtablationcumulativeuniqdata};
			\end{axis}
		\end{tikzpicture}
				\fi
	\caption{Cumulative unique valid red-black trees throughout sampling, for our type-based RBT generator tuned for the different objective functions.
	We regularize weights as described in \autoref{sec:regularization}.
	}
	\label{fig:unique1}
\end{figure}

To illustrate these objectives, we took our
type-based generator for color-labeled binary trees and tuned it for each
of them. Once we had the tuned weights, we sampled \(10^5\) trees
from the generator and computed the number of unique and valid RBTs we obtained
over sampling. We show the results in \autoref{fig:unique1}. 

When we tuned this generator for the entropy objective function, we got diverse color-labeled binary trees but not many that were valid with respect to the RBT invariant. On the other hand,
When we tuned it for validity we got valid test cases, but they are not diverse, as also shown in \autoref{fig:unique1}.
Indeed, that generator mostly produced red-black trees of height 0 in order to trivially satisfy the RBT invariant.
Finally,
\autoref{fig:unique1} shows that when tuned for specification entropy, the generator generates a much higher number of red-black trees that are both unique and valid, and it greatly outperforms the untuned version.

\subsubsection{Targeting Diverse, Valid Generations with Respect to a Feature}

In \autoref{sec:objectives-target}, we discussed how PBT users might wish to target a particular distribution over some feature of their generated test cases. What if they instead want to improve the diversity with respect to a feature? One can then tune the generator weights to maximize the entropy of the push forward of the generator distribution within the space of valid test cases. We define it formally as follows:

\begin{definition}[feature specification entropy objective function]
	    Given a generator \(G\), a function \(g: T \rightarrow T'\), and a validity condition \(\phi\),
		the feature specification entropy objective function
		is defined as
	\begin{equation*}
		\mathit{FeatureSpecificationEntropy(\assignment)} \mathdefeq - \sum_{
			\substack{t \in T \\ g(t) = t' \\ \phi(t) = 1}} p_{G, \assignment}(t) \log{p_{G, \assignment, g}(t')}.
	\end{equation*}
\end{definition}

\begin{example}
	Consider that \(G\) is a type-based STLC generator, and one wishes to tune its weights to produce well-typed terms that are of diverse types. Here, $\phi$ is the validity condition of well-typedness, and $g$ is a function that takes as input an STLC term and outputs its type. Then we can tune the weights of $G$ using feature specification entropy. 
\end{example}
	
\autoref{fig:stlc-unique-types} shows how the distribution of the types of the generated STLC terms changes as the weights get tuned for both well-typedness and diversity of types. Recall from \autoref{sec:overview-general-objectives} that tuning decreased
the proportion of terms that are ill-typed or of type \mono{bool} from 91\% to less than 1\%. 

\section{\ld: A Language for Tunable Generators} \label{sec:loaded-dice}

The previous section described how we can map the intuition of PBT users to mathematical objectives. The users still need to write their generators whose weights can be tuned for these objective functions. For this purpose, we first describe our language \ld, where users can write their generators as probabilistic programs with symbolic weights. We then describe its implementation as an embedded DSL in Julia. Finally, we describe how users can add more tunable weights in their generators to increase their ability to be optimized for an objective function. 

\subsection{Syntax and Semantics}

\begin{figure}[t]
	\begin{align*}
    \defkind{Types}{\tau}{\assign}{\mathrm{Bool} \variantor \tau_1 \times \tau_2}
		\defkind{Values}{v}{\assign}{
			T
			\variantor F
			\variantor (v, v)
		}
		\defkind{Expressions}{\mathit{aexp}}{\assign}{
			x \variantor v
		}
		\defkind{}{e}{\assign}{
			\mathit{aexp}
			\variantor \uop{fst}{\mathit{aexp}}
			\variantor \uop{snd}{\mathit{aexp}}
			\variantor \letin{x}{e}\ e
			\variantor \uop{flip}{q}
		}
		\defkind{}{}{\variantor }{
			\dif{\mathit{aexp}}{e}{e}
		}
    \defkind{Program}{p}{\assign}{
      e %
    }
		\defkindnobr{Numeric Terms}{q}{\assign}{
			c \variantor \symweight 
    }
	\end{align*}%
	\caption{
		The syntax of \ld, which is a subset of \dice~\cite{dice} except that unlike \dice we allow symbolic weights.
		The metavariable $x$ ranges over variable names,
		$c$ ranges over real numbers in the range $[0, 1]$,
    $f$ ranges over function names,
    and
		$\theta$ ranges over variable names for symbolic weights.
	}\label{fig:dice-syntax}
\end{figure}

\newcommand{\dbracket}[1]{ \left\llbracket { #1 } \right\rrbracket} 
\newcommand{\mcrot}[4]{\multicolumn{#1}{#2}{\rlap{\rotatebox{#3}{#4}~}}} 
\newcommand{\Llet}[2]{ {\texttt{let}~#1~\texttt{in}~#2} } %
\newcommand{\Lobs}[1]{ {\texttt{observe}~#1}}
\newcommand{\Lflip}[1]{ {\texttt{flip}~#1}}
\newcommand{\Lfst}[1]{ {\texttt{fst}~#1}}
\newcommand{\Lsnd}[1]{ {\texttt{snd}~#1}}
\newcommand{\Lite}[3]{ {\texttt{if}~#1~\texttt{then}~#2~\texttt{else}~#3}}
\newcommand{\te}[0]{ \texttt{e}}  %
\newcommand{\true}[0]{ \mathtt{T} }
\newcommand{\false}[0]{ \mathtt{F} }

\begin{figure}
	\centering
		\begin{align*}
			\dbracket{v_1}(v) \defeq \big(\delta(v_1)\big)(v)
			\quad\quad \dbracket{\Lfst{(v_1, v_2)}}(v) \defeq \big(\delta(v_1)\big)(v)
			\quad\quad \dbracket{\Lsnd{(v_1, v_2)}}(v) \defeq \big(\delta(v_2)\big)(v)
		\end{align*}
		\begin{align*}
			\dbracket{\Lite{v_g}{\te_1}{\te_2}}(v) \defeq
			\begin{cases}
				\dbracket{\te_1}(v) ~& \text{if } v_g = \true\\  
				\dbracket{\te_2}(v) ~& \text{if } v_g = \false\\
				0 \quad& \text{otherwise}
			\end{cases} \quad\quad
			\dbracket{\Lflip{c}}(v) \defeq& \begin{cases}
				c ~& \text{if }v = \true\\
				1-c ~& \text{if }v=\false\\
				0 ~& \text{otherwise}
			\end{cases}
		\end{align*}
		\begin{align*}
			\dbracket{\Llet{x = \te_1}{\te_2}}(v)\defeq
			\sum_{v'}\dbracket{\te_1}(v') \times \dbracket{\te_2[x \mapsto v']}(v) 
		\end{align*}
		\caption{Semantics for \dice{} expressions. The function $\delta(v)$ is a
			probability distribution that assigns a probability of 1 to the value $v$ and 0
			to all other values.}
		\label{fig:semantics}
	\end{figure}

For automated generator tuning, we treat generators as probabilistic programs that represent distributions over test cases. For the purpose of writing generators, we describe \ld, an extension of the discrete probabilistic programming language \dice~\cite{dice} with symbolic weights.

The syntax of \ld is given in \autoref{fig:dice-syntax}. \ld is a first-order functional language with support for booleans, tuples, and typical operations over these types.\footnote{
	\ld also provides support for \emph{probabilistic conditioning}, but we omit it in \autoref{fig:dice-syntax} because none of the generators require it.} It is augmented with the ability to create Bernoulli distributions via the \mono{flip} syntax: the expression \mono{flip $q$} represents a distribution that is true with probability $q$ and false with probability $1 - q$.
In \dice, arguments to \mono{flip}s must be numeric constants; in \ld, they may also be {\em symbolic weights} denoted by metavariable $\symweight$.
During generator tuning, it is these very symbolic weights that are tuned to optimize for an objective function. We describe the process of tuning in \autoref{sec:training-techniques}.

\ld inherits its semantics from \dice~\cite{dice}, replicated in \autoref{fig:semantics}.
For \dice, the semantic function \(\dbracket{\cdot}\)
maps expressions to
probability distributions,
where these probability distributions are functions from values to their probability mass. 
In order to support symbolic weights in \ld, we lift these semantics
to the semantic function \(\dbracket{\cdot}_L\), which
maps expressions $e$ with $n$ symbolic weights
and assignments $w \in [0,1]^n$ 
to the distribution that \dice semantics would result in if all symbolic
weights were substituted by $\assignment$. Formally,
\[\dbracket{e}_L(\assignment) \defeq \dbracket{e[\symweight \mapsto \assignment]}.\]
The \ld semantics provide us a representation of the distribution in terms of the
symbolic weights, which we require to learn weights as described in \autoref{sec:training-techniques}.

\subsection{Implementation} \label{sec:enriching-ld}

\newcommand{\bm}[1]{\textbf{\mono{#1}}}
\begin{figure}[htb]
	\begin{align*}
    \defkind{Types}{\tau}{\assign}{
      \textit{t} \variantor \mono{Int} \variantor \mono{Bool} \variantor 
      \mono{Tuple}\mono{\{}\tau_1,\dots\mono{\}}
      }
		\defkind{Statements}{s}{\assign}{
      \text{\textbf{\mono{@type}}}~\textit{t} = C_1(\tau_{11}, \dots)~\text{``\textbar{}''}~C_2(\tau_{21}, \dots)~\text{``\textbar{}''}~\dots
		}
    \defkind{Numeric Terms}{q}{\assign}{
       c \variantor \theta
    }
		\defkind{Expressions}{e}{\assign}{
      \bm{@match } e \mono{ (} C_1(x_{11}, \dots) \to e_1\mono{, }C_2(x_{21}, \dots) \to e_2, \dots \mono{)}}
   \defkind{}{}{}{
     \variantor \bm{@dice if } e_1\ e_2 \bm{ else } e_3 \bm{ end}
     \variantor \mono{flip(} q \mono{)} \variantor \mono{(} e_1, \dots \mono{)}
   }
    \defkindnobr{}{}{}{
      \variantor \mono{freq(} q_1 \Rightarrow e_1, \dots \mono{)}
      \variantor \mono{backtrack(} q_1 \Rightarrow e_1, \dots \mono{)}
    }
  \end{align*}
  \caption{Syntax for the library functions and macros that make up the Julia embedding of \ld. The metavariable $x$ ranges over variable names, $C$ ranges over constructor names, $t$ ranges over type names, $c$ ranges over numeric constants, and $\theta$ ranges over symbolic weights.}
  \label{fig:julia-syntax}
\end{figure}

We implement \ld as an embedded domain-specific language (DSL) in Julia, and include extensions that make it easier to express the kinds of generators used by practitioners. 

We show the syntax for the library functions and macros that make up the Julia
embedding of \ld in \autoref{fig:julia-syntax}.
As shown in earlier examples, this embedding includes the ability to declare algebraic data types and to pattern match on them. We provide the standard \mono{freq} and \mono{backtrack}\footnote{\mono{backtrack} samples from a list of optional values, resampling without replacement upon sampling \mono{None}.} combinators for PBT~\cite{PierceSF4}.  
There are also constructs that correspond one-to-one with \ld, such as \mono{\textbf{@dice if}}, \mono{flip}s, and tuples.
These constructs in the embedded language form a library of functions and macros which are composed within a larger Julia program.

To be precise, this Julia program is not a generator, but a metaprogram whose
execution produces a \ld program in the
syntax shown in \autoref{fig:dice-syntax}.
Thus, the functions and macros of the embedded DSL (\autoref{fig:julia-syntax})
are a library for constructing the data structure representing \ld programs.
The rest of the program is ``just Julia''
-- there is no special handling of Julia's language constructs, which are simply executed to produce the \ld program. This allows one to use arbitrary constructs (such as loops, side effects, etc.) to aid in constructing a generator, without them being materialized in the final \ld program.

In order to implement data structures such as integers and inductive types, 
we 
\emph{bit-blast}~\cite{bruttomesso09scalable,bitblasting}
them to representations in terms of \ld's tuples and booleans.
For example, integers are binary-encoded as tuples of booleans~\cite{scaling-integer-arithmetic}. 
To represent algebraic data types (e.g. lists and trees), we encode them as sum types
and use an explicit discriminator and placeholder values for absent components of the sum. Specifically,
the sum type \(\tau_1 + \tau_2\) is encoded as the nested product type \(\mathrm{Bool} \times \tau_1 \times \tau_2\). Here, the boolean indicates whether the value is of type \(\tau_1\) or \(\tau_2\) and the other two components encode the value.
This can be generalized to sums of arbitrary numbers of types by using an integer value as the discriminator.
For example, for lists of length up to two, \mono{[10;20]} is encoded as \mono{(2,10,(2,20,(1,)))} and \mono{[10]} is encoded as \mono{(2,10,(1,0,(0,)))}, where \mono{2} is the tag for \mono{Cons}, \mono{1} is the tag for \mono{Nil}, and \mono{0} is a placeholder value for absent arguments.

Lowering expressions to the core language is also straightforward.
Pattern  matching on values of an algebraic data type is lowered to the use of \mono{if} expressions, as is standard.  The \mono{freq} and \mono{backtrack} combinators are lowered to \mono{if} and \mono{flip}.
For example, \mono{freq([$q_1 \Rightarrow e_1$,$q_2 \Rightarrow e_2$,$q_3 \Rightarrow e_3$])} is lowered to
\mono{\textbf{if} flip($\frac{q_1}{q_1 + q_2 + q_3}$) \textbf{then} $e_1$ \textbf{else if} flip($\frac{q_2}{q_2 + q_3}$) \textbf{then} $e_2$ \textbf{else} $e_3$ \textbf{end}}.

\section{Constructing Tunable Generators} \label{sec:constructing-tunable-generators}

Using \ld,
users can define the structure of a generator and 
leave the weights undetermined, to be optimized by
automated generator tuning.
However, the space of distributions that are possible to achieve by tuning the
weights of the generator is limited by the generator's structure.\footnote{Using
the notation from \autoref{sec:objectives}, a generator \(G\) with \(n\) symbolic parameters exhibits the space of distributions represented by $\{ p_{G,w} \mid
\assignment \in [0,1]^n \}$.}
This, in turn, limits the extent to which \ld can optimize for an objective
function.\footnote{This is the classic problem of \textit{underfitting}, which
is well-studied in the machine learning literature~\cite{pml1Book}.}
We describe how generators can be written to be more ``tunable,''
and how we can derive tunable generators from a type definition.

\subsection{Adding Dependencies} \label{sec:good-generators-to-tune}

\newcommand{\maxsize}{\ensuremath{m}\xspace}
\newcommand{\codesfwidth}{0.96\textwidth}
\newcommand{\codespaceabovecaption}{\vspace{-0.5em}}

\begin{figure}[p]
\newcommand{\continuationnum}{3}
\begin{subfigure}{\codesfwidth}
	\begin{lstlisting}
X$\label{line-gen-tree-flips-red}$XgenColor() = @dice if flip($\theta_\mono{red}$) R() else B() end

genTree(size) =
  @match size (
    0 X$\to$X Leaf(),
    S(n) X$\to$X
      X$\label{line-gen-tree-flips-leaf}$X@dice if flip(X$\theta_\mono{leaf}$X)
        Leaf()
      else
        Branch(genColor(), genTree(n), genNat(), genNat(), genTree(n))
      end)

G = genTree(5)
	\end{lstlisting}
  \codespaceabovecaption
	\caption{An RBT generator following the same structure as QuickChick's type-based generators.}\label{fig:gen-tree-flips}
\vspace{1em}
\end{subfigure}
\begin{subfigure}{ \codesfwidth }
	\begin{lstlisting}[firstnumber=\continuationnum]
genTree(size) =
  @match size (
    0       X$\to$X Leaf(),
    S(n) X$\to$X (
      X\hl{w = \textbf{@match} size (1 $\to$ $\theta_1$, ..., \maxsize $\to$ $\theta_{\maxsize}$);}X
      @dice if flip(X\hl{w}X)
        Leaf()
      else
        Branch(genColor(), genTree(n), genNat(), genNat(), genTree(n))
      end))
	\end{lstlisting}
  \codespaceabovecaption
	\caption{Adding weights that depend on \mono{size} increases the distributions over height the generator can express.
	}\label{fig:gen-tree-flips-split-by-size}
\vspace{1em}
\end{subfigure}
\begin{subfigure}{ \codesfwidth }
	\begin{lstlisting}[firstnumber=\continuationnum]
genTree(size, X\hl{parentColor}X) =
  @match size (
  	0       X$\to$X Leaf(),
  	S(n) X$\to$X (
  	  X\hl{w = \textbf{@match} (size, parentColor) ((0,\textbf{R}()) $\to$ $\theta_\mono{0R}$, (0,\textbf{B}()) $\to$ $\theta_\mono{0B}$, $\dots$, (\maxsize,\textbf{B}()) $\to$ $\theta_\mono{\maxsize{}B}$);}X
  	  @dice if flip(X\hl{w}X)
        Leaf()
  	  else
        X\hl{c = genColor()}X
        Branch(X\hl{c}X, genTree(n, X\hl{c}X), genNat(), genNat(), genTree(n, X\hl{c}X))
  	  end))
			\end{lstlisting}
  \codespaceabovecaption
	\caption{Adding a function parameter allows weights to depend on both \mono{size} and the color of the parent call.
	}\label{fig:gen-tree-flips-parent-color}
\vspace{1em}
\end{subfigure}
\begin{subfigure}{ \codesfwidth }
	\begin{lstlisting}[firstnumber=\continuationnum]
genTree(size, X\hl{leaf}X) =
  @match size (
    0       X$\to$X Leaf(),
    S(n) X$\to$X
      X$\label{line-gen-tree-flips:9}$X@dice if X\hl{leaf}X
        Leaf()
      else
        X\hl{leftLeaf, rightLeaf = freq([ $\theta_\mono{1}$$\to$(\textbf{F},\textbf{F}), $\theta_\mono{2}$$\to$(\textbf{F},\textbf{T}), $\theta_\mono{3}$$\to$(\textbf{T},\textbf{F}), $\theta_\mono{4}$$\to$(\textbf{T},\textbf{T}) ]);}X
        Branch(genColor(), genTree(n, X\hl{leftLeaf}X), genNat(), genNat(), genTree(n, X\hl{rightLeaf}X))
      end)
	\end{lstlisting}
  \codespaceabovecaption
	\caption{Restructuring the generator to frontload choices allows correlations between choices to be introduced.
	}\label{fig:gen-tree-flips-is-leaf}
\end{subfigure}
\caption{Modifications to a generator for RBT trees (not necessarily valid ones)
to increase the space of distributions it can exhibit, written with syntactic sugar in \ld.
}
\label{fig:adding-dependencies-techniques}
\end{figure}

We describe general techniques to make generators more amenable to automated tuning.

\subsubsection{Parameterizing Weights by Function Arguments}\label{sec:splitting}

Consider the generator in \autoref{fig:gen-tree-flips}.
The user may wish to tune it for a particular distribution over heights.
However, the generators' distribution over heights only depends on one symbolic variable, $\theta_\mono{leaf}$, which limits the 
extent to which tuning can optimize this generator. 

A simple way to increase the expressivity of a generator is to add weights that depend on information already in scope.
Concretely, 
rather than using $\theta_\mono{leaf}$ in all invocations of \mono{genTree},
we can select one of multiple weights based on the current value of the \mono{size} parameter, as shown in \autoref{fig:gen-tree-flips-split-by-size}.
Let \maxsize denote the maximum size that this function is called with (5, in \autoref{fig:gen-tree-flips}).
Note that this generator now uses \(\maxsize + 1\) symbolic weights instead of only two.

To add more symbolic weights in their generator, the user does not have to be limited by the preexisting structure of their generator.
They can add additional function arguments to their generators for more symbolic weights.
For example, in the RBT generator, the user can pass down the chosen color to each subcall, in order to parameterize the symbolic weights by the color of the parent node.
This is shown in \autoref{fig:gen-tree-flips-parent-color}.
Now, the generator consists of \(2\maxsize + 1\) tunable weights.

\subsubsection{Correlating Random Choices in the Generator}

We described how we can increase the number of symbolic weights in a generator
by parameterizing them over the function arguments. This technique, including
the change shown in \autoref{fig:gen-tree-flips-parent-color}, has another effect: it correlates random choices in
the generator that were previously independent. In particular, the generator in
\autoref{fig:gen-tree-flips} chooses between different constructors, namely
\texttt{Leaf} and \texttt{Branch}, independently of the color of the parent
node.
But that is no longer the case in \autoref{fig:gen-tree-flips-parent-color}, which parametrizes symbolic weights by the \texttt{parentColor}. Depending on the parent \texttt{color}, the probability of choosing \texttt{Leaf} changes.

The resulting dependency among random choices allows the generator to more closely fit an objective function.
One can add more dependencies by \emph{frontloading} random choices, allowing them to be made in tandem. For example, the RBT generator in \autoref{fig:gen-tree-flips} makes independent choices for the constructors of the two children of a \texttt{Branch}. Instead, frontloading these choices can correlate them, as shown in \autoref{fig:gen-tree-flips-is-leaf}. The \texttt{genTree} function in that figure chooses the constructors for the two children beforehand and passes that information as arguments to the subcalls.

\subsection{Automatically Deriving Generators with Dependencies}
\label{sec:explain-derived}

\quickcheck provides its users the convenience of deriving simple generators for types automatically from their definitions. We observe that we can do the same for algebraic data types in \ld, but we can additionally use the ideas from above to make these generators more ``tunable'' by introducing additional dependencies.

\paragraph{API} To automatically derive generators with additional dependencies, we provide the metaprogramming~\cite{bawden99quasiquotation} function \mono{derive_generator}. It takes as input an algebraic data type definition in \ld, an integer representing the maximum size of generated values, and an integer specifying the \emph{stack lookback length}. Given this information, \mono{derive_generator} produces a generator for the given type that is parameterized by a size value and a portion of the execution context. Specifically, the tracked execution context is a representation of a suffix of the call stack,
up to the specified lookback length.
For example, for a binary tree data type, a stack lookback length of 2 indicates that the choice of a node should depend on the execution trace since the recursive call corresponding to the node two levels higher.
The generator also frontloads choices so that they can be correlated, as shown in the previous subsection, so it also passes down choices to the appropriate places.

\paragraph{Implementation} At a high level, the function \mono{derive_generator} constructs the desired generator in \ld by introducing new symbolic weights for each possible value of the dependencies, specifically the current size and portion of the call stack. To provide more intuition, we illustrate how \mono{derive_generator} works using the example of our type-derived RBT generator in \autoref{fig:derived-rbt-generator}, where the user has provided a maximum size of 5 and a lookback window length of 2.

The function \mono{derive_generator} first defines a type representing a choice of constructor (Lines \ref{line:colorc} and \ref{line:treec} of \autoref{fig:derived-rbt-generator}).
It then defines a sized generator with three arguments: \mono{size}, \mono{stack} and \mono{chosenCtor}. The derived generator here first checks if \mono{size} is zero, in which case it simply generates a \mono{Leaf}. If \mono{size} is not zero, it  pattern matches on \mono{chosenCtor}, which represents the choice of which constructor to use at this node in the tree, and then makes the random choices for the arguments to that constructor.

The code in \autoref{fig:derived-rbt-generator} uses a helper function that we have created called \mono{freqDep}.
Like the \mono{freq} combinator, \mono{freqDep}
makes a weighted random choice from a list of values. But while the \mono{freq} combinator specifies the weights directly, \mono{freqDep} takes \textit{dependencies} as an additional argument, and it introduces a different set of symbolic weights per value of the dependencies.  
In line \ref{line:freq-dep-example} of \autoref{fig:derived-rbt-generator}, we use \mono{freqDep} to choose among the eight possible combinations of (left-subtree constructor, color, right-subtree constructor) for the \mono{Branch} node being created, but with different symbolic weights per value of \mono{deps}, which includes the current size and relevant portion of the call stack.  The \mono{freqDep} function is implemented as a sequence of conditionals that branches on the possible values\footnote{This set of possible values is a static value with respect to the \ld program, and is computed in the lowering of \mono{freqDep}. %
} of \mono{deps} and introduces separate symbolic weights for each one.

A more general treatment of \mono{derive_generator} can be found in Appendix A. %

\definecolor{cmtcolor}{rgb}{0.5,0.6,0.5}
\begin{figure}[htb]
\begin{lstlisting}[
  commentstyle=\color{cmtcolor},
breaklines=true,
emph={@match,@type,function,for,in,let,match,with,end,type,and,as,begin,fun,of,map,mapi,product,*,|,(,)},
escapeinside=XX,
mathescape=false,
literate={->}{{$\rightarrow$}}2,
columns=flexible
]
@type Color = Red | Black
@type Tree = Leaf | Branch(Tree,Color,Int,Tree)

# Types automatically generated to represent the choice of constructor
@type ColorC = RedC | BlackCX\label{line:colorc}X
@type TreeC = LeafC | BranchCX\label{line:treec}X

genColorHelper(size, stack, chosenCtor) = # elided for brevity

function genTreeHelper(size, stack, chosenCtor)
  deps = (size, stack, chosenCtor)
  @match size (
    0 ->X\ XLeaf,
    S(size') ->
      @match chosenCtor (
        LeafC ->X\ XLeaf,X\label{line:leafc}X
        BranchC ->X\ Xbegin
          argChoices = freqDep(deps, [X\label{line:freq-dep-example}X
            # We use () as the dummy value for the choice of the constructors for
            # the int argument, as it is a builtin type.
            [LeafC,RedC,(),LeafC], [LeafC,RedC,(),BranchC],
            [LeafC,BlackC,(),LeafC], [LeafC,BlackC,(),BranchC],
            [BranchC,RedC,(),LeafC], [BranchC,RedC,(),BranchC],
            [BranchC,BlackC,(),LeafC], [BranchC,BlackC,(),BranchC]])
          Branch(
            # 2 is the configured stack lookback length.
            # 0, 1, and 2 number the program locations where recursive calls are made.
            genTreeHelper(size', firstN(2, cons(0, stack)), argChoices[0]),
            genColorHelper(size', firstN(2, cons(1, stack)), argChoices[1]),
            genIntDep(deps),
            genTreeHelper(size', firstN(2, cons(2, stack)), argChoices[3]))
        end))
end

function genTree()
  rootCtor = freqDep((), [LeafC, BranchC])
  genTreeHelper(5, [], rootCtor) # 5 is the configured initial size
end
\end{lstlisting}
\caption{An automatically derived generator for red-black trees with dependencies (slightly simplified for presentation purposes).}
\label{fig:derived-rbt-generator}
\end{figure}

\section{Automatically Tuning a \ld Generator} \label{sec:training-techniques}

Given an objective function and a generator with symbolic weights,
the task of generator tuning is to determine the weights that maximize the objective function.
To achieve this, we use the typical optimization algorithm, gradient descent,
which requires computing gradients of the objective function
to determine the direction\footnote{Typically, gradient descent updates the parameters in the \textit{opposite} direction of the gradient, in order to \textit{minimize} a loss function rather than to maximize an objective function. Thus we technically are performing gradient \textit{ascent}.} in which to update the weights.

The na\"ive methods for both probabilistic inference and computing gradients 
require enumerating all execution paths in a discrete probabilistic program.
\dice scales \textit{inference} by exploiting program structure;
we leverage its existing compilation strategy to scale the computation of gradients. 
Still, given a method to compute gradients of \ld programs,
entropy-based objectives involve all possible test cases, and
enumerating the distribution is similarly intractable.
We instead adapt REFINFORCE, a gradient estimation technique,
to replace this enumeration with sampling.
Finally, we discuss regularization techniques to avoid overfitting generator weights.

\subsection{Probabilistic Inference and its Gradient}

Tuning generators via gradient descent requires computing the gradients of the objective function with respect to the symbolic weights. Since the objectives depend on the generator distribution (as described in \autoref{sec:objectives}), we need probabilistic inference as a primitive.

This poses the first key performance challenge: since probabilistic inference for arbitrary \dice programs is \#P-hard in general, computing its gradients is also a \#P-hard problem. \dice scales probabilistic inference
by compiling programs to data structures that
exploit program structure
to
compact the representation of distributions.
In \ld, we leverage the very same compilation procedure to scale the computation of gradients.

\subsubsection{Probabilistic Inference in \ld} \ld, as an extension of \dice,
inherits its strategy for probabilistic inference. \dice compiles its
probabilistic programs to ordered binary decision diagrams (OBDDs) and reduces
the task of probabilistic inference to computations on the compiled OBDD. The
task of probabilistic inference is \#P-hard in general, but OBDDs
exploit structure in probabilistic programs to reuse
intermediate computations and scale probabilistic inference when possible.

An OBDD produced by \dice and \ld is a directed acyclic graph where each node corresponds to a boolean random variable with a \emph{level} and is connected to two child nodes via a \emph{high edge} (\(h\)) and a \emph{low edge} (\(l\)).
The high edge corresponds to the variable being true and the low edge corresponds to the variable being false.
The two terminal nodes, \mono{T} and \mono{F}, don't have children.
Each level corresponds to a \texttt{flip} in the program and its associated probability (which may be symbolic in \ld).
Given fixed values for each boolean variable, one can traverse the OBDD by starting from its root and following the edges corresponding to the value of the each variable to evaluate the program as \mono{T} or \mono{F}.

Once the program is compiled to an OBDD, the task of probabilistic inference is reduced to a bottom-up traversal on this graph. Each node \(n\) in an OBDD is associated with the probability \(\mathit{pr}(n)\) of reaching the terminal \mono{T} from that node. Our goal is to compute \(\mathit{pr}(\texttt{root})\), %
which gets further used in the computation of the objective function. Now, \(\mathit{pr}(\texttt{root})\) can be computed recursively in a single bottom-up pass of the OBDD using the following equations, where \(w\) is the weight associated with the level of the node and \(h\) and \(l\) are its high and low children, respectively~\cite{probwmc, inference}.
\begin{align}
	\mathit{pr}(\mono{T}) &= 1,
	\;\;\;\;\;\;
	\mathit{pr}(\mono{F}) = 0 \tag{Base cases}
	\\ 
	\mathit{pr}(n) &= w\cdot\mathit{pr}(h) + (1 - w)\cdot\mathit{pr}(l) \tag{Inductive Case}
\end{align}

Note that the bottom-up traversal of an OBDD runs in time linear in the size of the OBDD~\cite{graph-based-boolean}. This implies that the task of probabilistic inference also runs in time linear in the size of the OBDD. Thus, if one can get a small OBDD for a probabilistic program, one can also achieve efficient probabilistic inference for the program.

As an example, consider the \ld program in \autoref{fig:ldprogram}. \ld compiles this program to the OBDD in \autoref{fig:bdd}. Note that even though there are 8 possible instantiations of the coin flips in the program, the OBDD consists of only three nodes. For this OBDD, \(\mathit{pr}(\text{root})\) can be computed as shown in \autoref{fig:wmc-program}, using the equations above.

\begin{figure}[t]
	\begin{subfigure}[b]{0.3\textwidth}
		\centering
		\begin{lstlisting}[emph={if,then,let,in}]
let x = flipX\(_1\)X X\(\theta_0\)X in
let y = if x then flipX\(_2\)X X\(\theta_0\)X 
        else flipX\(_3\)X 0.9 in
y
		\end{lstlisting}
		\caption{A \ld program.} \label{fig:ldprogram}
	\end{subfigure}
  \hfill
	\begin{subfigure}[b]{0.2\textwidth}
            \ifcharts
			\begin{tikzpicture}
			\def\lvl{20pt}
			\node (f1) at (0, 0) [bddnode] {$f_1$};
			\node (f2) at ($(f1) + (-20bp, -\lvl)$) [bddnode] {$f_2$};
			\node (f3) at ($(f1) + (20bp, -2*\lvl)$) [bddnode] {$f_3$};
			\node[draw, rounded corners] at ($(f1) + (-45bp, 0bp)$)  {\(\theta_0\)};
			\node[draw, rounded corners] at ($(f1) + (-45bp, -20bp)$)  {\(\theta_0\)};
			\node[draw, rounded corners] at ($(f1) + (-45bp, -40bp)$)  {0.9};
			\node (false) at ($(f3) + (-30bp, -1*\lvl)$) [bddterminal] {F};
			\node (true) at ($(f3) + (-15bp, -1*\lvl)$) [bddterminal] {T};
			\node (y) at ($(f1) + (0bp, 20bp)$) [bddroot] {$y$};
			\begin{scope}[on background layer]
				\draw [highedge] (f1) -- (f2);
				\draw [lowedge] (f1) -- (f3);
				\draw [highedge] (f2) -- (true);
				\draw [lowedge] (f2) -- (false);
				\draw [highedge] (f3) -- (true);
				\draw [lowedge] (f3) -- (false);
				\draw[-stealth] (y) -- (f1);
			\end{scope}
		    \end{tikzpicture}
            \fi
		\caption{Compiled BDD representation.}
		\label{fig:bdd}
	\end{subfigure}
  \hfill
  \begin{subfigure}[b]{0.33\textwidth}
\begin{flalign*}
	\mathit{pr}(f_3) &= 0.9 \times 1 + (1 - 0.9) \times 0 = 0.9,
	\\
	\mathit{pr}(f_2) &= \theta_0 \times 1 + (1 - \theta_0) \times 0 = \theta_0,
	\\
	\mathit{pr}(f_1) &= \theta_0 \times \mathit{pr}(f_2) + (1 - \theta_0) \times \mathit{pr}(f_3) 
	\\
	&= \theta_0^2 + 0.9 (1 -\theta_0).%
\end{flalign*}
    \caption{The weighted model counting computation for node probabilities.}
    \label{fig:wmc-program}
  \end{subfigure}
  \caption{A \ld program is first compiled to an OBDD which is then traversed via the procedure of weighted model counting (WMC). WMC generates an expression in terms of the symbolic weights, which can then be autodifferentiated to optimize via gradient descent.}
\end{figure}

\subsubsection{Computing Gradients} 

Now, we describe how we can leverage the structure-exploiting properties of an OBDD to compute gradients efficiently.
First, note that the expressions in \autoref{fig:wmc-program} are differentiable with respect to the symbolic weights. This is actually the case for exact probabilistic inference in \ld in general. As a result, probabilistic inference in \ld can support generator tuning via gradient descent.\footnote{Throughout this paper, we initialize weights to have uniform values.} 

The question that remains is how one can actually compute these gradients efficiently.
The standard method in machine learning libraries~\cite{pytorch, jax,
tensorflow} is to compute gradients using \emph{automatic differentiation} over a
\emph{computation graph}. The computation of these gradients scales linearly
with the size of the computation graph. So if the computation graph is small,
computing gradients is efficient.
In fact, the OBDD that was earlier used for inference is exactly the computation graph we compute gradients for. All we need are the partial derivatives of \(\mathit{pr}(n)\) for an OBDD node $n$,
which compose via the chain rule of differentiation to compute the overall gradient.
\begin{flalign*}
	 \frac{\partial{\mathit{pr}(n)}}{\partial{\theta}} &= (\mathit{pr(h)} - \mathit{pr(l)})
	\;\;\;\;\;\;\;\;\;
	 \frac{\partial{\mathit{pr}(n)}}{\partial{\mathit{pr}(h)}} = \theta
	\;\;\;\;\;\;\;\;\;
	 \frac{\partial{\mathit{pr}(n)}}{\partial{\mathit{pr}(l)}} = 1 - \theta
\end{flalign*} 

Particularly, in \ld, we use the above equations to implement reverse-mode differentiation~\cite{autodiff} over the compiled OBDDs.

To have a complete picture of the workflow, consider again the \ld program in \autoref{fig:ldprogram}. The \ld compiler first compiles it to an OBDD as shown in \autoref{fig:bdd}. Then, via the process of weighted model counting, the compiler produces code corresponding to the resulting probability distribution, equivalent to that shown in \autoref{fig:wmc-program}. Since objective functions are computed in terms of this resulting probability distribution, the generated code along with the code to compute the objective functions constitutes a symbolic loss function and can be autodifferentiated to obtain gradients for each weight in the \ld program. These gradients are then used to perform one update in the algorithm of gradient descent.

\subsection{Scaling Gradient Computation for Entropy-Based Objective Functions}\label{sec:scale-entropy}

Now that we can efficiently compute gradients of the objective function, we can use gradient descent to automatically tune the generators.
But a key performance challenge still remains: for entropy-based objectives (\autoref{sec:objectives-spec-entropy}), the objective functions themselves enumerate all possible test cases,
as they are expectations with respect to the generator distribution. 
For example, consider the entropy objective function below, written as an expectation.

\begin{equation*}
		\mathit{Entropy}(\theta) = H(p_{G, \theta}) = - \mathop{\mathbb{E}_{x \sim p_{G, \theta}(.)}} \log{p_{G, \theta}(x)} = - \sum_{x \in X} p_{G, \theta}(x) \log{p_{G, \theta}(x)}
\end{equation*} 
Computing exact gradients for this function requires differentiating inference for all possible test cases that the generator can produce, which is not amenable to scaling.
To scale this computation, one may approximate the objective function via Monte Carlo sampling~\cite{monte-carlo-book}, as expectations can be approximated by averaging the inner expression over $N$ samples.
The resulting computation graph looks like \autoref{fig:commute}.
However, to compute the gradients over this computation graph, one needs to differentiate through the sampling operation, but as a discrete operation, it is not differentiable.

\begin{figure}[h]
	\centering

\begin{tikzpicture}[every text node part/.style={align=center},  transform shape,
		scale=0.9,
		node distance=5.3cm,
		auto,
		thick,
		box/.style={rectangle, draw, rounded corners=3pt, minimum width=1.5cm, minimum height=0.8cm, align=center, inner sep=6pt},
		arrow/.style={->, >=stealth, thick}
  ]
    \node[box] (gen) {Generator\\Distribution};
    \node[box] (samples) [right of=gen] {Samples\\$x_i \sim p_{G, \theta}(.)$};
    \node[box] (expectation) [right of=samples] {$\mathop{\mathbb{E}_{x \sim p_{G, \theta}(.)}} f(x)$};
    
    \draw [-stealth] (gen.east) -- (samples.west) node[midway, above, yshift=0.1cm] {Sampling};
    \draw [-stealth] (samples.east) -- (expectation.west) node[midway, above, yshift=0.1cm] {Average $f(x_i)$};
\end{tikzpicture}

	\caption{The computation graph for approximating an expectation. To approximate $\mathbb{E}[f(x)]$,
  we can sample from the generator distribution and average $f(x)$ over the samples.  }
  \label{fig:commute}
\end{figure}

To resolve this issue, we adapt a well-known gradient estimation technique from the literature, REINFORCE~\cite{REINFORCE}, to our context of automated generator tuning. Specifically, for entropy based objective functions, we can estimate its gradient as follows. We include the derivation in \mbox{Appendix C}.
\begin{equation*}
	\nabla_{\theta} \mathop{\mathbb{E}}_{x \sim p_{G, \theta}(.)} [\log{p_{G, \theta}(x)}] = \mathop{\mathbb{E}}_{x \sim p_{G, \theta}(.)} [\log{p_{G, \theta}(x)} \nabla_{\theta} \log{p_{G, \theta}(x)}]
\end{equation*}
The equation above eliminates the need to differentiate through the sample operation. Instead it reformulates the gradient as another expectation, known as a \textit{gradient estimator}, which can be directly computed from the samples by the same process as \autoref{fig:commute}. 
Note that in the computation of the inner expression for each sample,
\ld is used to compute the contained gradient exactly.

\begin{figure}[b]
	\centering
	\begin{subfigure}[b]{0.7\textwidth}
		\ifcharts
		\begin{tikzpicture}
			\begin{axis}[
				height=2.2cm,width=3.5cm,
				grid=major,
				xlabel near ticks,
				ylabel near ticks,
				scale only axis,
				xticklabel style={
					/pgf/number format/fixed,
					/pgf/number format/1000 sep={,\!},
					/pgf/number format/precision=0,
					font=\scriptsize
				},
				xlabel style = {font=\footnotesize},
				ylabel style = {font=\footnotesize},
				title style = {font=\small},
				title={Cumulative Unique Valid RBTs},
				xlabel={Number of Samples},
				ylabel={\# of Unique Valid RBTs},
				yticklabel style = {font=\scriptsize},
				xminorticks=false,
				yminorticks=false,
				scaled x ticks=false,
				legend style={
					at={(1.1,0.5)},
					anchor=west,
					font=\scriptsize,
					draw=none
				},
				legend entries={Untuned,Specification Entropy Without Regularization,
					Specification Entropy,
				},
				]
				\addplot[mark=none, thick, \initialcolor] table[x=Samples, y=UntunedGenerator.v] {\rbtablationcumulativeuniqdata};
				\addplot[mark=none, thick, \trainedcolor, dash pattern=on 3pt off 3pt, dash phase=0pt] table[x=Samples, y=SEGenerator.v] {\rbtablationcumulativeuniqdata};
				\addplot[mark=none, thick, \trainedcolor] table[x=Samples, y=SEBoundGenerator.v] {\rbtablationcumulativeuniqdata};
			\end{axis}
		\end{tikzpicture}
		\fi
	\end{subfigure}
	\caption{
Cumulative unique valid red-black trees throughout sampling, for our type-based RBT generator tuned for \specentropy, with and without regularization via bounded weights.
	}
	\label{fig:unique}
\end{figure}

\subsection{Regularization to Avoid Overfitting} \label{sec:regularization}

While performing optimization using gradient descent, it is very common to run into the problem of overfitting. Automated generator tuning is no exception to this problem. Overfitting happens when the weights of the generator are over-optimized for the objective function. Particularly, when we tune generators for the specification entropy objective function, the generator avoids producing more diverse terms to avoid the penalty for producing an invalid term. 

To avoid overfitting, regularization turns out to be an effective technique. Typical regularization techniques include adding a \textit{penalty term} to the objective function or eliminating certain values for the parameters. For effective generator tuning to avoid overfitting, we employ the latter and bound the weights in the generators between \([0.1, 0.9]\).

To demonstrate the effectiveness of using a regularization technique, we tune a type-based generator for red-black trees for the specification entropy objective function with and without regularization. Once the generators are tuned, we sample \(10^5\) trees using these generators and record the number of unique and valid RBTs we obtain. It is clear from the results shown in \autoref{fig:unique} that tuning with regularization allow the generators to achieve a much higher number of unique valid red-black trees as we obtain more samples.

\section{Evaluation: Bug-Finding on \etna Benchmarks} \label{sec:etna}

Previous sections, and in particular \autoref{sec:overview}, have already
demonstrated that our approach gives developers better control over their
generators' distributions. The experiments in \autoref{fig:height-tunings},
\autoref{fig:stlc-unique-types}, \autoref{fig:unique1}, and \autoref{fig:unique}
show that tuned generators are successfully optimized for various
objective functions.

In this section, we show that the better control we afford developers is
actually useful for the core purpose of PBT. In other words, we ask:
\begin{center}
\textbf{How effective is \ourframework at improving bug-finding performance?}
\end{center}

To investigate this question, we implemented \ld as an embedded domain-specific
language in Julia~\cite{Julia-2017} and used it to test various case studies
from the \etna benchmark suite~\cite{etna}. \etna was specifically developed to
evaluate PBT tools on how quickly they find bugs (pre-placed by the
benchmark authors) in example programs.

\subsection{The Testing Workloads}

We evaluate our approach on three of the four \rocq workloads from the \etna benchmark suite. The three workloads are designed to evaluate different PBT approaches to test programs that take as inputs binary search trees (BST), red-black trees, and terms of the simply-typed lambda calculus (STLC). Each of these workloads has a set of generators as well as a set of ``tasks,'' or bugs intentionally planted in the programs.
We use our approach to tune generators based on the following strategies for these three workloads:

\begin{itemize}
    \item \textbf{STLC, BST, and RBT Type-Based Generators}: We automatically derived type-based generators with additional dependencies for these types, as described in \autoref{sec:explain-derived}, with a stack lookback length of 2. We then tuned these generators for diversity and validity via \specentropy using our approach. We compare these tuned generators against the untuned\footnote{As with the rest of this paper, an ``untuned'' generator is one in which each random choice is uniformly distributed.} versions as well as the original generators.

    \item \textbf{STLC Bespoke Generator}: We first adapted \etna's bespoke STLC generator,
	a handwritten generator for STLC terms that uses backtracking to always generate well-typed terms, 
	to \ld.
	We fixed initial sizes and parameterized weights by the size argument of the current recursive call (\autoref{sec:good-generators-to-tune}). The generator is shown in Appendix B.
    We then tuned it to leverage existing insight by the authors of \etna. In particular, they make the observation that larger generations can be empirically detrimental for bug-finding (in the face of conventional PBT wisdom). In accordance with this observation, we tuned the STLC Bespoke generator for a target distribution over the number of syntactic function applications (\texttt{App} constructors), i.e.\ $\{0 \to 40\%, 1 \to 30\%, 2 \to 20\%, 3 \to 10\%\}$. We compare the tuned generator against the untuned one as well as the original bespoke generator. 
\end{itemize}

\subsection{Methodology}
We evaluated each of the above generators for the coverage and speed of bug-finding. We report the time (median over 11 trials)
it takes each of these generators to find the bugs in their respective workloads in \autoref{fig:plots} and \autoref{fig:stlc-bespoke-etna}. We report a timeout for a particular bug at 60 seconds.
We also report the time it took to tune the weights of these generators in Appendix D.

\subsection{Results}
\newcommand{\bespokemidrule}{}

\input{timings}

\newcommand{\doublerow}[2]{ \begin{tabular}{@{}r@{}}#1\\#2\end{tabular}}

\begin{figure}[!t]
    \centering
	\hspace{1em}
	\begin{subfigure}[t]{0.31\textwidth}
            \ifcharts
    		\begin{tikzpicture}
    			\begin{axis}[
    				height=4cm,
    				width=5.3cm,
    				grid=major,
    				xmode=log,
    				xtick={0.00001, 0.0001, 0.001, 0.01, 0.1, 1, 10, 60},
    				xlabel style = {font=\small},
    				ylabel style = {font=\small},
    				ylabel={Number of Bugs},
    				xlabel={Time (in s)},
    				yticklabel style = {font=\scriptsize},
    				xticklabel style = {font=\scriptsize},
    				xminorticks=false,
    				yminorticks=false,
                    legend style={font=\scriptsize},
                    legend entries={\etna,Initial,Tuned},
                legend style={at={(0.98,0.02)}, anchor=south east, font=\scriptsize},
    				]
    				\addplot[mark=none, thick, teal, densely dashed, mark size = 1.0pt] table[x={BespokeGeneratorx}, y={BespokeGeneratory}] {\bespokestlcdata}; %
    				\addplot[mark=none, thick, \initialcolor, mark size = 1.0pt] table [x={LBespokeGeneratorx}, y={LBespokeGeneratory}] {\bespokestlcdata}; %
    				\addplot[mark=none, thick, \trainedcolor, mark size = 1.0pt] table [x={SBespokeMLENumAppsTarget4321LR1Epochs250Generatorx}, y={SBespokeMLENumAppsTarget4321LR1Epochs250Generatory}] {\bespokestlcdata}; %
    			\end{axis}
    		\end{tikzpicture}
            \fi
		\caption{Time in seconds vs. number of bugs found (higher is better). The x-axis uses log scaling.}
  \label{fig:stlc-bespoke-etna-timings}
  \end{subfigure}
    \hfill
    \begin{subfigure}[t]{0.6\textwidth}
        \ifcharts
            \begin{tikzpicture}
                \begin{axis}[
                height=4cm,width=\textwidth-18pt,
                grid=major,
                xlabel near ticks,
                ylabel near ticks,
    			yticklabel style = {font=\scriptsize},
    			xticklabel style = {font=\scriptsize},
                xlabel style = {font=\small},
                ylabel style = {font=\small},
                title style = {font=\scriptsize},
                xlabel={Number of Applications},
                ylabel={Probability},
                xticklabel style = {font=\scriptsize},
                yticklabel style = {font=\scriptsize},
                xminorticks=false,
                yminorticks=false,
                scaled x ticks=false,
                legend style={font=\scriptsize},
                legend entries={Initial,Tuned,Objective},
                bar width=2.5pt,
                ymin=-0.02,
                enlarge x limits=0.06,
                xmax=16,
                    ]

                    \addplot[ybar, bar shift=-2.5pt, fill=\initialcolor!70, draw=\initialcolor!80, area legend] table[x={val}, y={probability}, restrict expr to domain={x}{0:16}] {\appsinitial}; %
                    \addplot[ybar, fill=\trainedcolor!70, draw=\trainedcolor!80, area legend] table [x={val}, y={probability}, restrict expr to domain={x}{0:16}] {\appstrained}; %
                    \addplot[ybar, bar shift=2.5pt, fill=black!30, draw=black, area legend, pattern=north east lines] table [x={val}, y={probability}, restrict expr to domain={x}{0:16}] {\appstarget}; %
                \end{axis}
            \end{tikzpicture}
        \fi
        \caption{Distribution of the number of applications. Truncated for space; all omitted values (max: 31) have probability less than 0.005.}
        \label{fig:stlc-bespoke-etna-dist}
    \end{subfigure}
  
	\caption{The results from tuning a bespoke STLC generator to leverage insight from \citet{etna}.}\label{fig:stlc-bespoke-etna}
\end{figure}

We found that \ourframework significantly improved bug-finding speed in all four of our benchmarks, in comparison to both the untuned versions and the generators from \etna. For the type-based stategy, tuning for \specentropy increased the bug-finding speed of our generators by 3.1--7.4$\times$ over the untuned versions and the \quickchick generators (\autoref{fig:timings-table}). \autoref{fig:plots} additionally shows that the number of bugs found by the tuned generator is greater or equal to the others at all times. For the bespoke STLC generator, tuning for a target distribution over the number of applications increased bug finding speed by at least 1.9$\times$ over the untuned version and the original version in \etna (\autoref{fig:timings-table}, \autoref{fig:stlc-bespoke-etna-timings}). Additionally, \autoref{fig:stlc-bespoke-etna-dist} validates that the distribution of applications did change as intended by tuning.

Training time for our tuned generators ranged from three to eight minutes.
One may notice that minutes of training time is significant compared to seconds to find bugs. However, this is a one-time cost, whereas generators are frequently run in continuous integration, sometimes as frequently as every code change.
Thus, the training cost amortizes over multiple testing runs.
We provide further detail into the training costs for our tuned generators in Appendix D.

\begin{table}[htb]
    \caption{Relative bug-finding speed of tuned generators over the original \etna generators and the untuned generators with additional parameters.
    	We report the geometric mean of the speedup relative to each baseline for all tasks in the workload.
           Timing data are from the same evaluation runs as \autoref{fig:plots} and \autoref{fig:stlc-bespoke-etna-timings}. 
    }\label{fig:timings-table}
    \begin{tabular}{@{}lcc@{}}
            \toprule
            & \multicolumn{2}{c}{Relative Bug-Finding Speed} \\
 Generator \& Workload & \etna $\to$ Tuned & Initial $\to$ Tuned \\ \midrule
            BST Type-Based       & 3.5$\times$      & 5.4$\times$                             \\
            RBT Type-Based       & 5.8$\times$      & 7.4$\times$                             \\

            STLC Type-Based      & 4.3$\times$      & 3.1$\times$                   \\ \bespokemidrule
                        STLC Bespoke         & 2.3$\times$      & 1.9$\times$                             \\        
            \bottomrule 
    \end{tabular}
\end{table}

\subsection{Internal Evaluation: Cause of Bug-Finding Speedup}

While these results show that tuning significantly improves bug-finding
performance, it is important to understand the cause of these speedups.
We hypothesize that \specentropy improves bug-finding by increasing the number of valid and
unique samples, and that tuning the STLC Bespoke generator for smaller terms increases bug-finding speed because of the increased generation speed.
To confirm this hypothesis, we measure the change in generation speed and the change in number of unique and valid samples and report in \autoref{fig:gen-speed-and-valid}.

\subsubsection{Generation Speed}

Tuning generators for specification entropy increased their bug-finding speed
(3.1--7.4$\times$) significantly more than their generation speed (0.5--1.3$\times$),
indicating that changes in generation speed were \textit{not} the primary cause of their speedup.
Conversely, tuning the STLC bespoke generator for smaller generations increased
its generation speed (4.4$\times$) more than its bug-finding speed
(1.9$\times$), also as expected.

We also observe that adapting \etna generators to the initial \ld generators changed 
generation speed (0.6--1.3$\times$) significantly less than the bug-finding
speedup of using tuned \ld generators instead of \etna generators (2.3--5.8$\times$). This
indicates that the primary advantage of the tuned generators was caused by the
tuning itself.

\subsubsection{Valid and Unique Samples}

We see in \autoref{fig:gen-speed-and-valid} that tuning for specification entropy significantly increases the number of valid and unique samples (matching the results of our prior experiment in \autoref{fig:unique1}).
In contrast, tuning the STLC Bespoke generator for smaller generations \emph{decreases}
the number of valid and unique samples. This highlights the tradeoff to be made 
depending on the tuning objective --- the virtue of tuning is that it allows one to choose that tradeoff.

We also provide validity rates of each generator, which consistently increases with tuning for \specentropy, in Appendix E.

\begin{table}[!htb]
    \caption{Relative generation speeds, and number of unique, valid samples out of 100,000 samples, for generators evaluated in \etna.
    }\label{fig:gen-speed-and-valid}
    \begin{tabular}{@{}lcccc@{}}
      \toprule
 & \multicolumn{2}{c}{Relative Generation Speed} & \multicolumn{2}{c}{Unique, Valid Samples}\\
 Generator \& Workload & \etna $\to$ Initial & Initial $\to$ Tuned & Initial & Tuned \\ \midrule
 BST Type-Based                   & 0.6$\times$  & 0.9$\times$  &  2,592 & 14,387     \\
 RBT Type-Based                   & 0.6$\times$  & 1.3$\times$  &   767 &  2,394      \\
 STLC Type-Based                  & 0.8$\times$  & 0.5$\times$  &   771 & 11,181      \\ \bespokemidrule
 STLC Bespoke                     & 1.3$\times$  & 4.4$\times$  & 55,322 & 10,589     \\
 \bottomrule
    \end{tabular}
\end{table}

\section{Discussion} \label{sec:discussion}

The previous section demonstrates empirically that our approach for automatic generator tuning enables faster bug finding. In this section, we discuss the
larger context in which our approach fits.

\subsection{Flat and Statically-Bounded Generators} \label{sec:flat-and-bounded}

Probabilistic programming languages such as \ld restrict the language they support in order to make the task of exact inference tractable. In contrast, typical generator languages such as \quickcheck are much richer. As a result, not all generators are expressible in \ld. Specifically, \ld only supports first-order (flat) generators that are statically bounded.

However, it turns out that many practical generators naturally tend to be expressible in first-order. All of the generators that are automatically produced from type information by tools like DRaGeN~\cite{dragen} and generic-random~\cite{generic-random} fit this criteria, as do all specification-based generators derived from \rocq inductive relations in QuickChick~\cite{paraskevopoulouComputingCorrectlyInductive2022,leo-good-generators}.
Handwritten generators tend to fit this pattern too: all generators used in the
benchmark suite provided by the \etna{} evaluation tool are flat, including complex and highly-tuned generators for well-typed programs in the simply-typed lambda calculus and System F\textsubscript{<:}.

\subsection{Adaptive Sizing}

Some generator frameworks, such as \quickcheck, have a notion of \textit{sized} generators. These are functions that output a generator given a \texttt{size} parameter. This allows the framework to sample test cases by progressively increasing the size of the generator. This is called \textit{adaptive sizing}. 

Since \ld can only express statically-bounded generators, to use our approach the initial size has to be specified. In \autoref{sec:etna}, we fix the initial size of our generators at 4, 4, and 5 for BST, RBT and STLC workloads respectively and find that generators tuned with our approach can outperform adaptively-sized generators anyway (generators referred to as ``\etna'' in \autoref{fig:timings-table}, \autoref{fig:plots}, and \autoref{fig:stlc-bespoke-etna}).
\ld also supports modeling
adaptive initial sizes as a fixed distribution,
resulting in a 
generator that can be tuned as a proxy for the adaptively-sized version.

\section{Related Work} \label{sec:related}

We report on related work for both generator tuning and probabilistic programming.

\subsection{Tuning Generators}
There are a number of existing approaches to generator tuning for PBT.

\paragraph{Manual tuning}
Perhaps the simplest solution to generator tuning is to provide knobs
for the user to do it
manually.  Indeed, this was part of the motivation for the
original QuickCheck~\cite{claessenQuickCheckLightweightTool2000} and
QuickChick~\cite{paraskevopoulouFoundationalPropertyBasedTesting2015} generator
languages.  This approach is simple, but it has the limitations we
discuss in the introduction: tuning a generator by hand requires significant
experimentation, since it's hard to know the overall distribution that will be
produced by a particular collection of local weights. \ld{} offers an automatic
approach.

Besides language-based approaches, there are auxiliary tools aimed at easing the tuning process. In particular, Tyche~\cite{tyche}
is a visual user interface
that provides insights into a generator's current distribution. This approach is complementary to \ld,
as it could be used to confirm that an objective had the
intended effect, or to visually compare objectives.
In turn, automatic tuning could enhance such interfaces,
e.g., by allowing the user to click-and-drag to adjust
distributions.

\paragraph{Online tuning}
One alternative to manual tuning is tuning during the testing process.  The
Target system~\cite{loscherTargetedPropertybasedTesting2017} uses
hill-climbing and simulated annealing during
generation to maximize an objective.
RLCheck~\cite{reddyQuicklyGeneratingDiverse2020} tunes generators with
reinforcement learning, seeking out diverse and valid inputs. Choice
Gradient Sampling~\cite{goldsteinParsingRandomness2022} tunes online
by manipulating the generator representation itself.
ISLa~\cite{steinhofelInputInvariants2022} uses an SMT solver during generation
to improve the chance of finding valid inputs.

Online techniques allow the generation process to target objectives over time,
but it is hard to predict the impact that they will have on the ultimate
distribution of generated test cases. By contrast, our approach gives direct
control over that final distribution. Furthermore, online approaches usually do
a significant amount of work during generation, leading to relatively slow
sampling speeds.  A recent study by
\citet{goldsteinPropertyBasedTestingPractice2024} suggests that some users of
PBT run their properties in a very tight loop, testing their properties as often
as every time they save their code. In these cases, online tuning, whether that
means running a learning algorithm or calling an SMT solver, may waste precious
time that could be used finding bugs.

\paragraph{Tuning by construction}
Some existing techniques automatically derive PBT generators
from data types or inductive specifications. For example,
DRaGeN~\cite{dragen} computes weights based on insights from the
literature on branching processes, aiming to uniformly distribute data
constructors. DRaGeN is faster than \ld{} at deriving and
tuning generators, but is hard-coded for a particular distributional goal and
does not work for user-defined generators.
Similarly, \citet{leo-good-generators} derive effective
generators from inductive relations in the Rocq theorem prover.
This process produces specification-satisfying generators by construction, but 
requires more up-front effort from the user and does not support 
distributional requirements beyond validity.  

\subsection{Differentiating Discrete PPLs}

Automatic differentiation has been used in probabilistic programming systems before, but only to compute the gradients of likelihood functions in order to guide the probabilistic inference algorithms. At a high level, the key novelty of \ld is that it uses automatic differentiation to compute \textit{exact} gradients of probabilistic inference itself.  

\paragraph{Gradient-Based Inference Algorithms}
Probabilistic inference for arbitrary probabilistic programs is hard. This has led to inference algorithms informed by gradients of the likelihood functions. Notably, Hamiltonian Monte Carlo~\cite{hmc} chooses the next sample using the gradient of the likelihood of the current sample. Variational
inference~\cite{variationalInference, pyro} utilizes gradients of the likelihood to approximate the posterior distribution via a family of closed form distributions. On the other hand, \ld performs discrete probabilistic inference and then computes the exact gradient with respect to the weights. Both HMC and variational inference are limited to continuous probabilistic programs, whereas \ld provides support for discrete programs, which is necessary for the application to PBT for data structures such as lists and trees.

\paragraph{Alternate Formulations of Parameter Tuning} We cast the problem of generator tuning as one of gradient descent relative to an objective function. 
An alternative would be to cast the problem as a form of Bayesian inference, which learns a posterior distribution for a model's parameters given a set of observations about the data generated by the model. In this formulation, the desired objective function would be modeled as a set of observations.
However, this approach would require casting our entropy objective as an observation, but entropy is not an event that can be observed but is rather a property of the entire distribution. In principle, one can cast our target distribution objective as an observable event. But again, this would involve reasoning about the generator distribution in its entirety rather than a single value the generator can produce. Thus, computing the posterior distribution of the generator weights in this setting would be highly intractable.

Another alternate formulation would be to treat our objective functions as likelihoods and attempt to compute the maximum-a-posteriori (MAP) estimation for the parameters in the generator under consideration. Computing the MAP estimation is  an optimization problem that can be computed via gradient-based approaches, and so that would be equivalent to what our approach accomplishes~\cite{pml1Book}.

\paragraph{Gradient Estimators for Discrete Probabilistic Programs}

Recent work such as
ADEV~\cite{adev} and \texttt{StochasticAD.jl}~\cite{stochasticAD} also compute gradients of probabilistic inference but via sampling. They propose 
program transformations 
to produce gradient estimators for probabilistic programs.
These works offer different variance and scaling tradeoffs: while \ld computes exact, zero-variance and zero-bias gradients, they employ Monte Carlo sampling to approximate these gradients.

\paragraph{Learning in Probabilistic Logic Programming}
\ld compiles programs to binary decision diagrams and differentiates through
them to learn weights. Similar functionality is found in
probabilistic logic programming where systems like DeepProblog and
Scallop~\cite{deepproblog, scallop, problogLeastSquares} allow users
to provide first-order logical specifications with weights that can be
learned as outputs of neural networks. \ld differs from these
approaches as it supports more traditional programming constructs; we also provide specific learning objectives and associated algorithms to improve PBT.

\section{Conclusion and Future Work} \label{sec:future}

In this paper, we presented a novel framework for automatically tuning PBT generators. 
We described how different intuitions of users about their generator distributions can be mapped to objective functions. 
We presented a new PPL, \ld, to express generators with support for symbolic weights and parameter learning. 
We also described how automated generator tuning can be made feasible and demonstrated its benfits in enabling PBT generators to find bugs faster.

In the future, we hope to extend \ld to provide support beyond flat and statically-bounded generators. In particular, common PBT frameworks support nested generators that induce distributions over distributions. We hope to reduce them to flat generators via \textit{defunctionalization}. We also hope to provide support for adaptive sizing. 
Finally, we wish to explore how automated generated tuning,
by allowing distributions to be specified declaratively rather than operationally,
can enable more user-friendly interfaces and APIs for property-based testing.

\section{Data-Availability Statement}

The artifact for this paper consists of the implementation of \ld as an embedding in Julia and code to reproduce experiments and plots in Sections \ref{sec:overview}, \ref{sec:objectives}, \ref{sec:training-techniques} and \ref{sec:etna}. It is available on Zenodo~\cite{artifact}. 
\ld is also available as an open-source repository on GitHub at \\\url{https://github.com/Tractables/Alea.jl/tree/loaded-dice}.

\begin{acks}
	
We would like to thank Leonidas Lampropoulos for his support and guidance as well as Steven Holtzen and Zilei Shao for useful technical discussions. This work is
supported in part by the \grantsponsor{nsf}{National Science Foundation}{https://www.nsf.gov/} under grants \grantnum{nsf}{CCF-2220891}, {\em SHF: Medium: Usable Property-Based Testing, NSF \#2402449} and \grantnum{nsf}{IIS1943641}, the Victor Basili Postdoctoral Fellowship at the University of Maryland, DARPA ANSR, CODORD, and SAFRON programs under awards FA8750-23-2-0004, HR00112590089, and HR00112530141, and gifts from Adobe Research, Cisco Research, and Amazon. Approved for public release; distribution is unlimited.
\end{acks}

\bibliographystyle{ACM-Reference-Format}
\ifappendix
\clearpage
\appendix

\section{Automatically Deriving Generators with Dependencies} \label{app:automatic-generator}

In the metaprogram shown in \autoref{fig:derived-generators}, the function \mono{derive_generator} takes as input an inductive type definition, an initial size to bound the size of generated values, and a stack lookback window length, and outputs a generator for that type with additional weights and dependencies.\footnote{We assume this function is used for all user-defined generated types reachable from the type we wish to generate. For example, calling \mono{derive_generator} for \mono{Tree} will derive \mono{genTreeHelper}, which is defined in terms of \mono{genColorHelper}, so we assume that code produced by calling \mono{derive_generator} for \mono{Color} is also available.} To describe how it achieves this, we first describe some library functions and then dive into the structure of the metaprogram.

\subsection{Preliminaries}
To add weights parametrized by the execution context, we use the following library functions.
\begin{itemize}
  \item The \mono{freqDep} combinator, as explained in \autoref{sec:explain-derived}.
  \item For each built-in type \mono{T}, the function \mono{gen<T>Dep}, which is a generator for values of type \mono{T} with weights parameterized by dependencies.
  The function takes a dependency value used to parameterize weights.
	\item For each inductive type \mono{T}, the function \mono{genTerminal<T>}, which is a generator for values of type \mono{T} that only use \mono{T}'s non-recursive constructors.
  It is used in order to handle the zero-size case of a sized generator.
	For example, \mono{genTerminalTree}, the na\"ive generator for red-black trees, would only produce the non-recursive \mono{Leaf} constructor.
  The function takes a dependency value used to parameterize weights.
\end{itemize}

\definecolor{cmtcolor}{rgb}{0.5,0.6,0.5}
\begin{figure}[!htb]
\centering
\newcommand{\q}[1]{\textcolor{purple}{#1}}
\newcommand{\cmt}[1]{\textcolor{cmtcolor}{#1}}
\newcommand{\openq}{"}
\newcommand{\closeq}{"}
\begin{lstlisting}[
         commentstyle=\color{cmtcolor},
    breaklines=true,
    emph={@match,function,for,in,let,match,with,end,type,and,as,fun,of,map,mapi,product,*,|,(,)},
    escapeinside=XX,
    mathescape=false,
    literate={->}{{$\rightarrow$}}2,
    columns=flexible
]
type ty = Inductive of string * ctor list | Builtin of string
and ctor = string * ty list
# let rec tree = Inductive ("Tree", [("Leaf",[]); ("Branch",[tree;color;Builtin "Int";tree])])

let enum ctors = map ctors (fun (name, _) ->X\ Xname ^ X\q{"C"}X) # enum tree = ["LeafC"; "BranchC"]

let derive_generator (Inductive (ty_name, ctors)) initial_size stack_lookback =
 X\q{\openq\textbf{@type}}X $(ty_name ^ X\q{"C"}X) X\q{=}X $(map (enum ctors) (fun c -> X \q{"| "}X ^ c))X\label{line:enum}X
  # @type TreeC = | LeafC | BranchC

  X\tikz[remember picture,overlay]{\coordinate (lsttop);}\q{\textbf{function} gen}X$(ty_name)X\q{Helper(\hlC{size, stack,}\tikz[remember picture,overlay]{\coordinate (hlc1);} \hlB{chosenCtor}\tikz[remember picture,overlay]{\coordinate (hlb1);})}XX\label{line:gen-helper}X
    X\q{\hlC{deps = (size, stack, chosenCtor)}\tikz[remember picture,overlay]{\coordinate (hlc2);}}XX\label{line:deps}X
    X\q{\textbf{@match} size (}X
      X\q{0 $\to$ genTerminal}X$(name)X\q{(deps),}X X\label{line:gen-t-dep}X
      X\q{S(size') $\to$}X
        X\q{\hlB{\textbf{@match} chosenCtor}\tikz[remember picture,overlay]{\coordinate (hlb2);} (}X
          $(map ctors (fun (name, args) ->
            X\q{"}X$(name ^ X\q{"C"}X) X\q{$\to$ \textbf{begin}}X $( # BranchC X\cmt{$\to$}X begin
              let each_arg_choices =
                map args (fun arg ->
                  match arg with
                  | Inductive (_, arg_ctors) ->X\ Xenum arg_ctors
                  | Builtin _ ->X\ X[X\q{\openq()\closeq}X])
              in
              X\q{\hlB{argChoices}\tikz[remember picture,overlay]{\coordinate (hlb3);} = freqDep(\hlC{deps}\tikz[remember picture,overlay]{\coordinate (hlc3);},}X $(product each_arg_choices)X\q{)}XX\label{line:freq-dep}X
              # argChoices = freqDep(deps,X\cmt{$\{\texttt{LeafC},\texttt{BranchC}\}\times\{\texttt{RedC},\texttt{BlackC}\}\times\{()\}\times\{\texttt{LeafC},\texttt{BranchC}\}$}X)X\phantom{y}X
              $(name) X\q{(}X$( # Branch (
                mapi args (fun i arg ->
                  match arg with
                  | Builtin name ->
                    X\q{\openq gen}X$(name)X\q{Dep(\hlC{deps}\tikz[remember picture,overlay]{\coordinate (hlc4);}),\closeq}XX\label{line:gen-int-dep}X # genIntDep(deps),
                  | Inductive { name; _ } -> 
                    X\q{\openq gen}X$(name)X\q{Helper(size', firstN(}X$(stack_lookback)X\q{, cons(}X$(loc())X\q{,stack)),}XX\label{line:update-stack}X
                                        X\q{argChoices[}X$(i)X\q{]),\closeq}X
                    # genTreeHelper(size', firstN(2, cons(11,stack)), argChoices[3]),
                ))X\q{)}X
            X\q{\textbf{end}\closeq}X))X\q{))}X
  X\q{\textbf{end}}X

  X\q{\textbf{function} gen}X$(ty_name)X\q{()}XX\label{line:gen-t}X
    X\q{rootCtor = freqDep((), }X$(enum ctors)X\q{)}X
    X\q{gen}X$(ty_name)X\q{Helper(rootCtor, }X$(initial_size)X\q{, [])}X
  X\q{\textbf{end}\closeq}X
\end{lstlisting}

\begin{tikzpicture}[remember picture,overlay]
  \node[
    rectangle,
    draw=gray,
    fill=white,
    font=\small,
    inner sep=5pt,
    text height=1.5ex,
    text depth=0.25ex,
    anchor=north west,
  ] (textbox1) at ([xshift=2.347in,yshift=-1.08in]lsttop) {
    {\color[rgb]{\hlCtextcolor} $\bullet$ Dependencies are added to parameterize weights.}
  };

  \node[
    rectangle,
    draw=gray,
    fill=white,
    font=\small,
    inner sep=5pt,
    text height=1.5ex,
    text depth=0.25ex,
    anchor=north west
  ] (textbox2) at ([xshift=2.8in,yshift=-0.34in]lsttop) {
    {\color[rgb]{\hlBtextcolor} $\bullet$ The choice of constructor is frontloaded.}
  };
  \coordinate (texthlc) at ([xshift=7pt,yshift=-0.5em-4pt]textbox1.north west);
  \coordinate (texthlb) at ([xshift=7pt,yshift=-0.5em-4pt]textbox2.north west);

  \draw[dotted, thick, color={\hlCcolor}, -{Latex[length=1.5mm]}] (texthlc) -- ($(hlc1)+(0pt,-3pt)$);
  \draw[dotted, thick, color={\hlCcolor}, -{Latex[length=1.5mm]}] (texthlc) -- ($(hlc2)+(0pt,-2.5pt)$);
  \draw[dotted, thick, color={\hlCcolor}, -{Latex[length=1.5mm]}] (texthlc) -- ($(hlc3)+(0pt,5pt)$);
  \draw[dotted, thick, color={\hlCcolor}, -{Latex[length=1.5mm]}] (texthlc) -- ($(hlc4)+(0pt,5pt)$);
  
  \draw[dotted, thick, color={\hlBcolor}, -{Latex[length=1.5mm]}] (texthlb) -- ($(hlb1)+(0pt,-3pt)$);
  \draw[dotted, thick, color={\hlBcolor}, -{Latex[length=1.5mm]}] (texthlb) -- ($(hlb2)+(0pt,6pt)$);
  \draw[dotted, thick, color={\hlBcolor}, -{Latex[length=1.5mm]}] (texthlb) -- ($(hlb3)+(0pt,6pt)$);
\end{tikzpicture}

\caption{Pseudocode for a metaprogram~\cite{bawden99quasiquotation} that derives \ld generators with dependencies. The metaprogramming language (OCaml-like) appears in black, while the generated target language (the embedding of \ld in Julia) appears in red.
Quotes are used to produce the target language, while \mono{\$()} splices in an expression from the meta language, which produces either a string or a list of strings that is implicitly concatenated.
Comments show example instances of generated code for red-black trees.
The metaprogram derives generators for recursive inductive types, and would be used by calling \mono{derive_generator} for all user-defined inductive types.
}
\label{fig:derived-generators}
\end{figure}

\subsection{Execution Context as a Dependency}
All of the above functions take a value upon which to split weights; we bind this value to \mono{deps} on \autoref{line:deps}. One of these dependencies is \mono{size}, which is already present in sized generators. Another is \mono{stack}, which is a list of program locations of length at most \mono{stack_lookback}.
We update the stack on \autoref{line:update-stack}, by prepending a unique integer corresponding to the source location (generated by \mono{loc()}, which simply returns the next unused integer from a global counter), then truncating the list to be no longer than \mono{stack_lookback}. The final dependency is \mono{chosenCtor}, which is the frontloaded choice of the constructor to be used at the root.

\subsection{Metaprogram Structure} Using the library functions defined above, the function \mono{derive_generator} produces the type-derived generator with dependencies. It achieves this by creating \mono{gen<T>Helper} (\autoref{line:gen-helper}): a sized \mono{T} generator (as in \autoref{fig:gen-tree-flips}), modified to split weights by dependencies (as in \autoref{fig:gen-tree-flips-split-by-size}), pass additional execution context (as in \autoref{fig:gen-tree-flips-parent-color}), and frontload choices (as in \autoref{fig:gen-tree-flips-is-leaf}).

Concretely, the additional execution context consists of \mono{stack} and \mono{chosenCtor}. The \mono{stack} is the suffix of the call stack (particularly \mono{parentColor} in \autoref{fig:gen-tree-flips-parent-color}). The \mono{chosenCtor} is the frontloaded choice of the constructor to be used at the root.
Compared to \autoref{fig:gen-tree-flips-is-leaf}, we generalize the boolean \mono{leaf} to an enum corresponding to the constructors of the type, to support inductive types with more than two constructors. This enum is generated on \autoref{line:enum}.
To handle the 0 case of the sized generator, which should surely terminate, we use \mono{genTerminal<T>}.
Finally, \mono{gen<T>} wraps \mono{gen<T>Helper} by choosing the root
constructor, and passing in an empty stack and an initial size (\autoref{line:gen-t}).

\FloatBarrier

\section{STLC Bespoke Generator} \label{sec:stlc-bespoke}
We include the ``bespoke'' generator for simply-typed lambda calculus terms tuned in \autoref{fig:height-tunings} and \autoref{sec:etna}.
It corresponds to the \quickchick generator from \etna~\cite{etna}, which is approximately 60 lines of \rocq, translated to \ld and modified
with added weights that depend on size and fixed initial sizes.

\begin{lstlisting}
@type Typ = TBool() | TFun(Typ, Typ)
@type Expr = Var(Nat) | Bool(Boolean) | App(Expr, Expr) | Abs(Typ, Expr)

genVarX$'$X(ctx, t, p, r) =
  @match ctx (
  	Nil() X$\to$X r,
  	Cons(tX$'$X, ctxX$'$X) X$\to$X
      @dice if t == tX$'$X
        genVarX$'$X(ctxX$'$X, t, p + 1, Cons(p, r))
      else
        genVarX$'$X(ctxX$'$X, t, p + 1, r)
      end)
 
genZero(env, tau) =
  @match tau (
    TBool() X$\to$X Some(Bool(genBoolean())),
    TFun(t1,t2) X$\to$X
      bindOpt(
        genZero(Cons(t1, env), t2),
        e -> Some(Abs(t1, e))))

genTyp(s) =
  @match s (
    0 X$\to$X TBool(),
    S(sX$'$X) X$\to$X (
      wX$_\mono{bool}$X, wX$_\mono{fun}$X =
        @match s (
          1 X$\to$X (X$\theta_\mono{bool1}$X, X$\theta_\mono{bool1}$X),
          2 X$\to$X (X$\theta_\mono{bool2}$X, X$\theta_\mono{fun2}$X));
      freq [
        wX$_\mono{bool}$X X$\Rightarrow$X TBool(),
        wX$_\mono{fun}$X X$\Rightarrow$X (
          t1 = genTyp(sX$'$X);
          t2 = genTyp(sX$'$X);
          TFun(t1, t2)
      ]))

genExpr(env, tau, size) =
  @match size (
    0 X$\to$X (
      backtrack [
        X$\theta_\mono{0var}$X X$\Rightarrow$X oneOf([
          None(),
          map(
            x X$\to$X Some(Var(x)),
            genVarX$'$X(env, tau, 0, Nil())
          )
        ]),
        X$\theta_\mono{0zero}$X X$\Rightarrow$X genZero(env, tau)
      ],),
    S(n) X$\to$X (
      wX$_\mono{var}$X, wX$_\mono{app}$X, wX$_\mono{val}$X = @match size (1 X$\to$X (X$\theta_\mono{var1}$X,X$\theta_\mono{app1}$X,X$\theta_\mono{val1}$X), ..., 5 X$\to$X (X$\theta_\mono{var5}$X,X$\theta_\mono{app5}$X,X$\theta_\mono{val5}$X));
      backtrack [
        (wX$_\mono{var}$X,
          oneOf(
            None(),
            map(
              x X$\to$X Some(Var(x)), 
              genVarX$'$X(env, tau, 0, Nil())))),
        (wX$_\mono{app}$X, (
          argty = genTyp(2);
          bindOption(genExpr(env, TFun(argty, tau), n),
            e1 X$\to$X
              bindOption(genExpr(env,argty,n),
              e2 X$\to$X
                Some(App(e1,e2)))))),
        (wX$_\mono{val}$X,
          @match tau (
            TBool() X$\to$X Some(Bool(genBool())),
            TFun(t1, t2) X$\to$X
              bindOption(genExpr(cons(t1,env),t2,n),(e X$\to$X
                Some(Abs(t1,e))))))
      ]))

G = genExpr(Nil(), genTyp(2), 5)
\end{lstlisting}

\section{Adapting REINFORCE for Entropy Gradient Estimation}\label{appendix:gradient-estimation}

\begin{proposition} The gradient of the entropy of $p_{G,\theta}$, which can be expressed as the expectation of $\log{p_{G, \theta}(x)}$, can be estimated as follows:
	\[
		\nabla_{\theta} \mathop{\mathbb{E}}_{x \sim p_{G, \theta}(.)} [\log{p_{G, \theta}(x)}] = \mathop{\mathbb{E}}_{x \sim p_{G, \theta}(.)} [\log{p_{G, \theta}(x)} \nabla_{\theta} \log{p_{G, \theta}(x)}]
	\]
\end{proposition}
\begin{proof}
\begin{flalign*}
	\nabla_{\theta} \mathop{\mathbb{E}}_{x \sim p_{G, \theta}(.)} [\log{p_{G, \theta}(x)}] &=  \nabla_{\theta}  \sum_{x \in X} p_{G, \theta}(x) \cdot \log{p_{G, \theta}(x)} 
	\\
	&= \sum_{x \in X} \nabla_{\theta} p_{G, \theta}(x) \cdot \log{p_{G, \theta}(x)} \tag{Leibnitz Integral Rule}
	\\
	&= \sum_{x \in X} p_{G, \theta}(x) \cdot \nabla_{\theta} \log{p_{G, \theta}(x)} + \log{p_{G, \theta}(x)} \cdot \nabla_{\theta} p_{G, \theta}(x) \tag{Product Rule of Differentiation}
	\\
	&= \sum_{x \in X} \nabla_{\theta}p_{G, \theta}(x) + \log{p_{G, \theta}(x)} \cdot p_{G, \theta}(x) \nabla_{\theta} \log{p_{G, \theta}(x)} \tag{\(\nabla_{\theta} \log{p_{G, \theta}(x)} = \frac{\nabla_{\theta}p_{G, \theta}(x)}{p_{G, \theta}(x)}\)}
	\\
	&= \sum_{x \in X} \log{p_{G, \theta}(x)} \cdot p_{G, \theta}(x) \nabla_{\theta} \log{p_{G, \theta}(x)} \tag{\(\sum_{x \in X} \nabla_{\theta} p_{G, \theta}(x) = \nabla_{\theta} \sum_{x \in X} p_{G, \theta}(x) = \nabla_{\theta} 1 = 0\)}
	\\
	&= \mathop{\mathbb{E}}_{x \sim p_{G, \theta}(.)} [\log{p_{G, \theta}(x)} \nabla_{\theta} \log{p_{G, \theta}(x)}]
\end{flalign*}
\end{proof}

\section{Internal Evaluation: Training Costs}\label{sec:training-costs}

To provide further detail into the training costs for our tuned generators,
\autoref{fig:eval-num-params-and-time} shows the number of parameters and training times for each tuned generator.

\begin{table}[htb]
    \caption{
    Number of parameters and training time for each tuning in the evaluation.
    We believe that STLC Bespoke has the longest training time despite having
the fewest weights due to its complex backtracking control flow.
    }\label{fig:eval-num-params-and-time}
    \begin{tabular}{@{}lrr@{}}
      \toprule
 Generator \& Workload & \# of Params. & Training Time \\ \midrule
 BST Type-Based       & 100 & 3m   \\
 RBT Type-Based       & 132 & 3m   \\
 STLC Type-Based      & 796 & 7m   \\ \bespokemidrule
 STLC Bespoke         & 23 & 8m \\
 \bottomrule
    \end{tabular}
\end{table}

\section{Validity Rates and Overhead of Rejection Sampling for \SpecEntropy}\label{sec:overhead-of-rejection-sampling}
For the type-based generators tuned in \autoref{sec:etna}, 
\autoref{fig:eval-pct-valid} shows the percentage of samples that are valid before and after training. 
While we previously showed that tuning for specification entropy increases the \textit{diversity} of valid samples (\autoref{fig:gen-speed-and-valid}), we now see that it also increases the absolute validity rate.
  
Tangentially, we note that our computation of \specentropy performs rejection sampling, as only valid samples affect the gradient, and thus has computational performance proportional to the validity rate.
In our workloads, this has limited impact, as the lowest untuned validity rate is 35\%.
In cases of sparse validity conditions, one may employ more sophisticated sampling algorithms such as Markov Chain Monte Carlo or Sequential Monte Carlo~\cite{monte-carlo-book} to sample more effectively, which we leave for future work.

\begin{table}[htb]
    \caption{Percentage of valid samples out of 100,000 generations.}\label{fig:eval-pct-valid}
    \begin{tabular}{@{}l>{\raggedleft\arraybackslash}p{1.3cm}>{\raggedleft\arraybackslash}p{1.3cm}>{\raggedleft\arraybackslash}p{1.3cm}@{}}
      \toprule
 & \multicolumn{3}{r}{\hspace{0.3em}Percentage of Valid Samples}\\
 Generator \& Workload & \etna & Initial & Tuned \\ \midrule
 BST Type-Based       & 76\%  & 66\%     & 71\%   \\
 RBT Type-Based       & 74\%  & 63\%     & 79\%   \\
 STLC Type-Based      & 40\%  & 35\%     & 65\%   \\ \bespokemidrule
 STLC Bespoke         & 100\% & 100\%    & 100\%  \\
 \bottomrule
    \end{tabular}
\end{table}

\clearpage

\bibliography{references}
\else

\bibliography{references}
\fi

\end{document}

%% file: abstract.tex
Property-based testing validates software against an executable specification by
evaluating it on randomly generated inputs.
The standard way that PBT users generate test inputs is via \textit{generators} that describe how to sample test inputs through random choices.
To achieve a good distribution over test inputs, users must \textit{tune} their generators,
i.e., decide on the weights of these individual random choices.
Unfortunately, it is very difficult to understand how to choose individual generator weights in order to achieve a desired distribution,
so today this process is tedious and limits the distributions that can be practically achieved.

In this paper, we develop techniques for the automatic and offline tuning of generators. Given a generator with undetermined \emph{symbolic weights} and an \emph{objective function}, our approach automatically learns values for these weights that optimize for the objective. We describe useful objective functions that allow users to (1) target desired distributions and (2) improve the diversity and validity of their test cases. We have implemented our approach in a novel discrete probabilistic programming system, \ld, that supports differentiation and parameter learning, and use it as a language for generators. We empirically demonstrate that our approach is effective at optimizing generator distributions according to the specified objective functions. We also perform a thorough evaluation on PBT benchmarks, demonstrating that, when automatically tuned for diversity and validity, the generators exhibit
a 3.1--7.4$\times$ speedup in bug finding.

%% file: timings.tex
\begin{figure}[!t]
	\newcommand{\subfigurewidth}{0.31\textwidth}
	\newcommand{\axiswidth}{5.3cm}
	\newcommand{\axisheight}{4cm}
	\hspace*{10pt}
	\begin{subfigure}[t]{0.9\linewidth}
		\centering
		\hspace{0.07\textwidth}
            \ifcharts
		\begin{tikzpicture} 
			\begin{axis}[%
				hide axis,
				xmin=10,
				xmax=50,
				ymin=0,
				ymax=0.4,
				legend columns=-1,
				legend style={draw=none,legend cell align=left,column sep=1ex}
				]
				\addlegendimage{mark=none, thick, teal, densely dashed, mark size = 1.0pt}
				\addlegendentry{\small \etna};
				\addlegendimage{mark=none, thick, \initialcolor, mark size = 1.0pt}
				\addlegendentry{\small Initial};
				\addlegendimage{mark=none, thick, \trainedcolor, mark size = 1.0pt}
				\addlegendentry{\small Tuned};
			\end{axis}
		\end{tikzpicture}
            \fi
	\end{subfigure}

    \begin{minipage}{\textwidth}
    \centering
            \raisebox{1.2cm}{%
                \begin{tikzpicture}
                    \node[rotate=90] at (0,0) {\small Number of Bugs};
                \end{tikzpicture}%
            }%
	\begin{subfigure}[t]{\subfigurewidth}
        \ifcharts
		\begin{tikzpicture}
			\begin{axis}[
				height=\axisheight,
				width=\axiswidth,
				grid=major,
				xmode=log,
				ymode=linear,
				xmax=65,
				xlabel style = {font=\small},
				ylabel style = {font=\small},
				xlabel={Time (in s)},
                title style = {font=\small},
                title={BST},
				legend columns=1,
				yticklabel style = {font=\scriptsize},
				xticklabel style = {font=\scriptsize},
				legend style={at={(5em,-1.5em)},anchor=north west, font=\tiny},
				xminorticks=false,
				yminorticks=false
				]
				\addplot[mark=none, thick, teal, densely dashed, mark size = 1.0pt] table[x={TypeBasedGeneratorx}, y={TypeBasedGeneratory}] {\bstdata}; %
				\addplot[mark=none, thick, \initialcolor, mark size = 1.0pt] table [x={
					BLSDThinGeneratorx
				}, y={
					BLSDThinGeneratory
				}] {\bstdata}; %
				\addplot[mark=none, thick, \trainedcolor, mark size = 1.0pt] table [x={
					BLSDThinSEFreq2SPB200IsBSTLR30Epochs2000Bound10Generatorx	
				}, y={
					BLSDThinSEFreq2SPB200IsBSTLR30Epochs2000Bound10Generatory
				}] {\bstdata}; %
			\end{axis}
		\end{tikzpicture}
        \fi
	\end{subfigure}
	\begin{subfigure}[t]{\subfigurewidth}
            \ifcharts
		\begin{tikzpicture}
			\begin{axis}[
				height=\axisheight,
				width=\axiswidth,
				grid=major,
				xmode=log,
				xtick={0.0001, 0.01, 1, 60},
				xmax=65,
				xlabel style = {font=\small},
				ylabel style = {font=\small},
				xlabel={Time (in s)},
                title style = {font=\small},
                title={RBT},
				yticklabel style = {font=\scriptsize},
				xticklabel style = {font=\scriptsize},
				xminorticks=false,
				yminorticks=false
				]
				\addplot[mark=none, thick, teal, densely dashed, mark size = 1.0pt] table[x={TypeBasedGeneratorx}, y={TypeBasedGeneratory}] {\rbtdata}; %
				\addplot[mark=none, thick, \initialcolor, mark size = 1.0pt] table [x={
					RLSDThinGeneratorx	
				}, y={
					RLSDThinGeneratory	
				}] {\rbtdata}; %
				\addplot[mark=none, thick, \trainedcolor, mark size = 1.0pt] table [x={
					RLSDThinSEFreq2SPB200IsRBTLR30Epochs2000Bound10Generatorx	
				}, y={
					RLSDThinSEFreq2SPB200IsRBTLR30Epochs2000Bound10Generatory
				}] {\rbtdata}; %
			\end{axis}
		\end{tikzpicture}
            \fi
	\end{subfigure}
	\begin{subfigure}[b]{\subfigurewidth}
            \ifcharts
		\begin{tikzpicture}
			\begin{axis}[
				height=\axisheight,
				width=\axiswidth,
				grid=major,
				xmode=log,
				xtick={0.0001, 0.01, 1, 60},
				xmax=65,
				xlabel style = {font=\small},
				xlabel style = {font=\small},
				ylabel style = {font=\small},
				xlabel={Time (in s)},
                title style = {font=\small},
                title={STLC},
				yticklabel style = {font=\scriptsize},
				xticklabel style = {font=\scriptsize},
				xminorticks=false,
				yminorticks=false
				]
				\addplot[mark=none, thick, teal, densely dashed, mark size = 1.0pt] table[x={TypeBasedGeneratorx}, y={TypeBasedGeneratory}] {\stlcdata}; %
				\addplot[mark=none, thick, \initialcolor, mark size = 1.0pt] table [x={SLSDThinGeneratorx}, y={SLSDThinGeneratory}] {\stlcdata}; %
				\addplot[mark=none, thick, \trainedcolor, mark size = 1.0pt] table [x={SLSDThinSEFreq2SPB200WellTypedLR30Epochs2000Bound10Generatorx}, y={SLSDThinSEFreq2SPB200WellTypedLR30Epochs2000Bound10Generatory}] {\stlcdata}; %
			\end{axis}
		\end{tikzpicture}
            \fi
	\end{subfigure}
    \end{minipage}
	\caption{Time in seconds vs. number of bugs found (higher is better) across
		workloads and generator strategies. The x-axes use log scaling.
		}\label{fig:plots}
\end{figure}